\documentclass[journal,onecolumn]{IEEEtran}

\usepackage{cite}
\usepackage{amsthm}
\usepackage{amsmath,amssymb,amsfonts}
\usepackage{graphicx}
\usepackage{algorithm,algorithmic}
\usepackage{hyperref}
\hypersetup{hidelinks=true}
\usepackage{textcomp}

\usepackage[T1]{fontenc}
\usepackage{bm}
\usepackage{color}
\usepackage{subfigure}
\usepackage{ulem}
\usepackage{array,makecell,booktabs}
\usepackage{threeparttable}
\usepackage{multirow}
\usepackage{listings}
\newcommand{\oomit}[1]{}
\newcommand{\argmax}{\mathop{\mathrm{argmax}}}

\newtheorem{assumption}{Assumption}
\newtheorem{definition}{Definition}
\newtheorem{proposition}{Proposition}
\newtheorem{corollary}{Corollary}
\newtheorem{lemma}{Lemma}
\newtheorem{remark}{Remark}
\newtheorem{theorem}{Theorem}
\newtheorem{example}{Example}
\newtheorem{problem}{Problem}

\def\BibTeX{{\rm B\kern-.05em{\sc i\kern-.025em b}\kern-.08em
    T\kern-.1667em\lower.7ex\hbox{E}\kern-.125emX}}
\begin{document}

\title{Controlled Reach-avoid Set Computation for Discrete-time Polynomial Systems via Convex Optimization}

\author{Taoran Wu, Yiling Xue, Dejin Ren, Arvind Easwaran,\\ Martin Fränzle, and Bai Xue
 \thanks{Taoran Wu, Yiling Xue, Dejin Ren, and Bai Xue are with Key laboratory of System Software (Chinese Academy of Sciences) and State Key Laboratory of Computer Sciences, Institute of Software, Chinese Academy of Sciences, Beijing, China 100190, and University of Chinese Academy of Sciences, Beijing, China 100049. Email: (\{wutr,xueyl,rendj,xuebai\}@ios.ac.cn).}
\thanks{Yiling Xue is also with School of Advanced Interdisciplinary Sciences, University of Chinese Academy of Sciences, Beijing, China 100190.}
\thanks{Arvind Easwaran is with College of Computing and Data Science, Nanyang Technological University, Singapore (e-mail: arvinde@ntu.edu.sg).}
\thanks{Martin Fränzle is with Carl von Ossietzky Universität, Oldenburg, Germany (e-mail: martin.fraenzle@uni-oldenburg.de).}
 }
 
\maketitle

\begin{abstract}
This paper addresses the computation of controlled reach-avoid sets (CRASs) for discrete-time polynomial systems subject to control inputs. A CRAS is a set encompassing initial states from which there exist control inputs driving the system into a target set while avoiding unsafe sets. However, efficiently computing CRASs remains an open problem, especially for discrete-time systems. In this paper, we propose a novel framework for computing CRASs which takes advantage of a probabilistic perspective. This framework transforms the fundamentally nonlinear problem of computing CRASs into a computationally tractable convex optimization problem. By regarding control inputs as disturbances obeying certain probability distributions, a CRAS can be equivalently treated as a $0$-reach-avoid set in the probabilistic sense, which consists of initial states from which the probability of eventually entering the target set while remaining within the safe set is greater than zero. Thus, we can employ the convex optimization method of computing $0$-reach-avoid sets to estimate CRASs. Furthermore, inspired by the $\epsilon$-greedy strategy widely used in reinforcement learning, we propose an approach that iteratively updates the aforementioned probability distributions imposed on control inputs to compute larger CRASs. We demonstrate the effectiveness of the proposed method on extensive examples.
\end{abstract}


\section{Introduction}
\label{sec:intro}
In today's interconnected world, discrete-time control systems are ubiquitous, spanning a wide range of applications from consumer electronics to medical devices. These systems are designed to seamlessly interface with digital computers and control physical processes at discrete time intervals \cite{alur2015principles}. Among these systems, polynomial discrete-time systems stand out as a versatile and widely used model\oomit{of dynamical system}. A polynomial discrete-time system is a type of discrete-time system that is described by a set of polynomial equations. Being able to model dynamical systems with linear as well as non-linear characteristics, polynomial discrete-time systems can accurately capture the complex behavior of real-world systems, making them an essential tool for engineers and researchers. 

A central problem in the analysis of dynamical systems that is equally important to engineers and researchers is the determination of reach-avoid properties of dynamical systems featuring controllable inputs. The question of interest here is whether it is possible, using adequate sequences of control inputs, to guarantee that a system can reach both a desired target state set and avoid an unsafe state set. For example, in motion planning, a robot may need to navigate through a cluttered environment to reach a goal location while avoiding obstacles. Similarly, in surveillance systems, it is crucial to detect and track targets while avoiding false alarms and ensuring the safety of humans. Controlled reach-avoid sets (CRASs) provide a rigorous and unified framework for achieving such objectives. They are a set of initial states that, when started from, guarantee that there are control inputs making the system reach the desired goal set while simultaneously avoiding the unsafe set. In other words, CRASs define the initial conditions under which a system can be safely, i.e.\ under avoidance of hazardous scenarios, steered towards a desirable target state. The present work focusses on the computation of CRASs for discrete-time polynomial systems.

Computing CRASs is notoriously challenging, but several methods have been suggested that can be adapted to polynomial discrete-time systems. These methods include reachability analysis \cite{gao2020computing,schafer2023scalable}, dynamic programming \cite{margellos2011hamilton,zhao2022outer}, and Lyapunov-like methods \cite{tan2006nonlinear,wang2018permissive,ren2024iterative}. Among these approaches, Lyapunov-like methods stand out as particularly promising techniques and have been gaining significant attention. These methods were originally developed for analyzing the stability of equilibrium points in dynamical systems, transforming the stability problem into the identification of a Lyapunov function. This transformation effectively circumvents the need for analytical solutions to the system. In recent decades, advancements in polynomial optimization, especially in sum-of-squares (SOS) polynomial optimization \cite{parrilo2000structured}, have further enhanced the automation and thus the potential of Lyapunov-like methods, making them powerful tools for computing various sets, including controlled domains of attraction and controlled invariant sets. When the system of interest is polynomial and lacks control inputs, the problem of finding Lyapunov-like functions can be addressed by solving a convex optimization problem constructed using the SOS decomposition technique for multivariate polynomials. Later, Lyapunov-like functions, such as control barrier functions \cite{ames2019control} and control guidance-barrier functions \cite{xue2024reach}, have been proposed for computing controlled invariant sets and CRASs. The synthesis of Lyapunov-like functions for systems with control inputs is inherently nonlinear, posing a significant challenge. However, this issue can be circumvented by converting the problem into a family of convex optimization problems, especially when the system is continuous-time and described by polynomial ordinary differential equations that are affine in the control inputs. In such cases, the nonlinear problem can be reformulated as a bilinear SOS program, which can be efficiently solved iteratively through convex optimization techniques \cite{tan2006nonlinear,clark2021verification,clark2022semi,wang2023safety,ren2024iterative}. 
Despite these advances, a significant gap remains in the existing literature and established methodologies regarding the synthesis of Lyapunov-like functions for discrete-time control systems using SOS polynomial optimization. The primary reason for this gap is the inclusion of nonlinear terms in the composition of Lyapunov-like functions and control inputs in existing Lyapunov-like conditions, which are more complex than bilinear terms. These nonlinear terms make the aforementioned SOS polynomial optimization approaches unsuitable for application to Lyapunov-like functions for discrete-time polynomial systems subject to control inputs\footnote{To this end, note that for discrete-time polynomial systems, affine control inputs and arbitrary polynomial inputs share the same expressiveness such that in contrast to the continuous-time setting, tractability cannot be gained from restricting the input type. The reason is that in discrete time dynamical systems, additional state variables with discrete-time polynomial dynamics can simulate polynomial control inputs by polynomially mapping affine inputs.}. In this work, we will bridge this gap and address this issue.

In this paper, we present a convex optimization framework to compute CRASs for discrete-time polynomial systems, which falls under the category of SOS optimization. Our method eliminates the nonlinear terms in the composition of discrete-time Lyapunov-like functions and control inputs by imposing a probability distribution on the control variables and applying an expectation operator to the composition of these functions and inputs. Specifically, we reformulate CRASs as $0$-reach-avoid sets, as demonstrated in \cite{xue2021reach}, in a probabilistic sense by assigning a probability distribution to control variables. This reformulation enables the application of the SOS polynomial optimization techniques previously developed for computing $0$-reach-avoid sets to compute CRASs. Moreover, we iteratively refine the initial probability distribution assigned to control variables using the $\epsilon$-greedy strategy, a common approach in reinforcement learning, to enhance the refinement of Lyapunov-like functions. This iterative refinement procedure, in which a convex optimization is solved in each iteration, allows for the continuous exploration of CRASs. Finally, we demonstrate the effectiveness of the proposed approach through extensive experiments and compare it with state-of-the-art methods.
 

The main contributions of this work are summarized below.
\begin{enumerate}
    \item  We approach the computation of CRASs for discrete-time polynomial systems from a probabilistic standpoint by assigning a defined probability distribution to the control variables. We show that CRASs are equivalent to $0$-reach-avoid sets in a probabilistic sense. Thus, the problem of computing CRASs can be equivalently reformulated as the computation of $0$-reach-avoid sets, thus facilitating the computation of CRASs by mitigating the nonlinear terms arising from the composition of Lyapunov-like functions and control inputs.
    \item  To achieve a less conservative CRAS, we introduce a novel iterative algorithm, which is inspired by the $\epsilon$-greedy approach commonly employed in reinforcement learning, and updates the probability distribution over control variables in each iteration to refine Lyapunov-like functions. 
    \item We demonstrate the effectiveness of the proposed methods through extensive experiments and provide a comprehensive comparison with existing methods of computing Lyapunov-like functions for discrete-time polynomial systems subject to control inputs.
\end{enumerate}

\subsection*{Related Work}
\label{sec:related}
To the best of our knowledge, there are no methods specially tailored for computing CRASs for discrete-time polynomial systems, although there are works on computing a set covering the maximal CRAS \cite{han2018controller,zhao2022outer} or reach-avoid sets for systems without control input \cite{xue2020inner}. Some other existing methods could be adapted for computing CRASs, such as 
dynamic programming methods \cite{margellos2011hamilton,zhao2022outer} and Lyapunov-like methods \cite{tan2006nonlinear,wang2018permissive,ren2024iterative}. Consequently, the related work reviewed below will also cover these methods, although they are not comprehensive.

Dynamic programming methods commonly necessitate the discretization of state and control spaces, a process that can significantly increase computation time exponentially with the system size. In contrast, Lyapunov-like methods, which involve the search for a Lyapunov-like function by solving an optimization problem, do not require such discretization, thereby circumventing the inherent exponential complexity inherent in state and control space discretization. Consequently, they have garnered increasing attention over the past decades, giving rise to a substantial body of computational methods such as symbolic computation methods \cite{jirstrand1997nonlinear}, counterexample-guided learning methods \cite{chang2019neural,ravanbakhsh2019learning,zhao2021learning,yang2023hybrid,vzikelic2023learning,chatterjee2023learner,abate2020formal}, and SOS programming methods \cite{tan2006nonlinear,wang2018permissive,ren2024iterative}. 

Symbolic computation methods for calculating Lyapunov-like functions are resource-intensive and exhibit a doubly exponential worst-case complexity \cite{davenport1988real}, which can severely impact the efficiency and practicality of these methods in various applications. Despite their ability to yield exact and theoretically sound results, the high computation time and memory requirements can make them impractical for many real-world scenarios. 

In recent decades, the field of polynomial optimization, particularly the Sum of Squares (SOS) approach, has become a pivotal tool for the computation of polynomial Lyapunov functions in control systems. For polynomial continuous-time control-affine systems, the inclusion of bilinear constraints makes SOS polynomial optimization inherently non-convex. Despite this challenge, the non-convex optimization can be effectively tackled through an iterative process that alternately updates Lyapunov-like functions and control functions, solving a convex optimization problem at each iteration. This strategy for decomposing bilinear SOS polynomial optimization into convex sub-problems is widely used in numerous existing works, including \cite{tan2006nonlinear,toulkani2024reducing,wang2018permissive,wang2023safety,ren2024iterative}. Prior to this study, the computation of Lyapunov-like functions for discrete-time control systems using SOS polynomial optimization has been significantly less explored compared to its continuous-time counterpart. The methods designed for continuous-time dynamics are not directly applicable to discrete-time systems due to the presence of nonlinear terms in the composition of discrete-time Lyapunov-like functions with control inputs, which are inherent to existing Lyapunov-like conditions. The aim of this research is to address this gap by proposing efficient SOS polynomial optimization techniques tailored for the discrete-time setting.

Recently, counterexample-guided learning methods for the computation of Lyapunov-like functions have garnered increasing attention within the formal methods and control theory communities. This is attributed to the development of powerful machine learning algorithms and their ability to deal with general nonlinear systems beyond polynomial ones. The counterexample-guided learning method for computing Lyapunov-like functions is an iterative process that includes two main components in each iteration: a learner and a verifier. The learner's role is to generate candidate Lyapunov-like functions, while the verifier assesses the validity of these candidates using SMT \cite{chang2019neural,zhao2021learning,edwards2023general}, SOS polynomial optimization \cite{yang2023hybrid}, mixed-integer
linear programming \cite{wu2023neural}, and other techniques such as the boundary propagation technique \cite{wang2024simultaneous}. If the verifier identifies a candidate function as an invalid one, it produces counterexamples --- specific scenarios where the candidate function fails to perform as intended. These counterexamples offer valuable feedback, enabling the learner to refine its function generation process. However, the learning and verification procedures in each iteration can be computationally intensive, and this method does not guarantee the successful identification of valid Lyapunov-like functions. In contrast to the aforementioned methods, the approach proposed in this work specifically targets discrete-time polynomial systems and focuses on identifying Lyapunov-like functions through convex optimization, which enhances computational efficiency. This advantage is illustrated through several numerical examples involving discrete-time polynomial systems, showcasing the effectiveness of the search for Lyapunov-like functions.

\subsection*{Paper Outline and Notations}
This paper is structured as follows: in Section \ref{sec:preliminaries}, we introduce the problem of computing CRASs of interest and discuss the challenges of existing Lyapunov-like conditions for characterizing CRASs. In Section \ref{sec:ra}, we introduce a method for computing the initial CRASs. Then our algorithm to expand CRASs is proposed in Section \ref{sec:update}. Section \ref{sec:exp} provides experimental evaluations, and conclusions are drawn in Section \ref{sec:conclusion}.

Throughout this paper, we refer to several basic notions. For instance, $\mathbb{N}$ and $\mathbb{R}^n$ denote the set of non-negative integers and n-dimensional real vectors, respectively; vectors are denoted by boldface lowercase letters; for an $m$-dimensional vector $\bm{z}$, $\bm{z}(i)$ denotes its $i$-th element, where $1\leq i\leq m$; for a vector $\bm{z}$ and a scalar $\delta$, $\bm{z}+\delta$ means $\delta$ is added to each component of the vector $\bm{z}$; for a state $\bm{x}\in\mathbb{R}^n$ and a constant $\epsilon > 0$, $\mathcal{B}(\bm{x}, \epsilon)$ represents the $\epsilon-$neighborhood of $\bm{x}$, i.e., $\mathcal{B}(\bm{x}, \epsilon) = \{\bm{y}\in\mathbb{R}^n \mid \|\bm{y}-\bm{x}\|<\epsilon\}$; $\mathbb{R}[\cdot]$ denotes the ring of polynomials with real-valued coefficients in variables given by the argument; $\sum[\bm{x}]$ denotes the set of sum of squares polynomials, i.e.,
$\sum[\bm{x}]=\{p(\bm{x})\mid p(\bm{x})=\sum_{i=1}^k q_i^2(\bm{x}), q_i(\bm{x})\in \mathbb{R}[\bm{x}]\}$.

\section{Preliminaries}
\label{sec:preliminaries}

In this section, we introduce discrete-time polynomial systems subject to control inputs, CRASs, and Lyapunov-like conditions for characterizing  CRASs. 

\subsection{Problem Formulation}
Consider a discrete-time polynomial system subject to control inputs, which is modeled by a 
polynomial state update equation of the form
\begin{equation}
    \begin{split}
        & \bm{x}(t+1) = \bm{f}(\bm{x}(t), \bm{u}(t)), \forall t\in \mathbb{N},\\
        & \bm{x}(0) = \bm{x}_0, \\
    \end{split}
    \label{eq:systems}
\end{equation}
where $\bm{x}(\cdot):\mathbb{N}\rightarrow\mathbb{R}^n$ is the dynamic state, $\bm{x}_0 \in \mathbb{R}^n$ is the initial state, $\bm{u}(\cdot):\mathbb{N}\rightarrow \mathcal{U}$ is the control input with $\mathcal{U} \subseteq \mathbb{R}^m$ being a compact set, and $\bm{f}(\cdot):\mathbb{R}^n\times \mathcal{U} \rightarrow \mathbb{R}^n$ with $\bm{f}(\cdot) \in \mathbb{R}[\cdot]$ is the system dynamics. In this paper, we consider the case where the input set $\mathcal{U}$ is a compact set of interval form, i.e., $\mathcal{U} = \{\bm{u}\mid\underline{\bm{u}} \leq \bm{u} \leq \overline{\bm{u}} \}$ where $\underline{\bm{u}}$ and $\overline{\bm{u}}$ are $m$-dimensional vectors. A control policy for system \eqref{eq:systems} is an $\omega$-sequence of control inputs, represented by a function $\bm{u}(\cdot): \mathbb{N} \rightarrow \mathcal{U}$.



Given an initial state $\bm{x}_0 \in \mathcal{D}$ and a control policy $\pi = \{\bm{u}(t)\}_{t\in\mathbb{N}}$, we can obtain a trajectory of system \eqref{eq:systems}, which is a sequence $\{\phi^\pi_{\bm{x}_0}(t)\}_{t\in\mathbb{N}}$ satisfying $\phi^\pi_{\bm{x}_0}(0) = \bm{x}_0$ and 
\begin{equation}
\label{eq:traj}
    \phi^\pi_{\bm{x}_0}(t+1)=\bm{f}(\phi^\pi_{\bm{x}_0}(t), \bm{u}(t)), \forall t \in \mathbb{N}.
\end{equation}


Let $h(\bm{x})$ and $g(\bm{x})$ be polynomial functions over $\bm{x}\in \mathbb{R}^n$ such that $h$ and $g$ are radially unbounded, i.e., $\lim_{\|\bm{x}\|\rightarrow +\infty} h(\bm{x})=+\infty\text{~and~}\lim_{\|\bm{x}\|\rightarrow +\infty} g(\bm{x})=+\infty$. For a polynomial, this condition is satisfied if the homogeneous polynomial consisting of its highest-degree monomials is positive definite \cite{ahmadi2017sum}.
Given these functions, we define a bounded open safe set $\mathcal{X} \subseteq \mathbb{R}^n$ and a bounded open target set $\mathcal{T} \subseteq \mathcal{X}$, where
\begin{equation*}
    \mathcal{X}=\{\bm{x}\in \mathbb{R}^n \mid h(\bm{x})< 0\},~\mathcal{T}=\{\bm{x}\in \mathbb{R}^n \mid g(\bm{x})< 0\}.
\end{equation*}
Then, a CRAS is defined as follows.

\begin{definition}[Controlled Reach-avoid Set]
    A CRAS $\mathcal{R}$ for system \eqref{eq:systems} is a set $\mathcal{R} \subseteq \mathcal{X}$ of initial states such that starting from any $\bm{x}_0 \in \mathcal{R}$, there exists at least one control policy $\pi$ under which the resulting trajectory hits the target set $\mathcal{T}$ in a finite time $k \in \mathbb{N}$ while staying inside the safe set $\mathcal{X}$ for all times $t\leq k$, i.e., 
    \begin{equation*}
      \exists \pi. \exists k\in \mathbb{N}.~[ \bm{\phi}_{\bm{x}_0}^{\pi}(k)\in \mathcal{T} \wedge \bigwedge_{t=0}^k \bm{\phi}_{\bm{x}_0}^{\pi}(t)\in \mathcal{X}]
    \end{equation*}
\end{definition}

As mentioned in the introduction, the computation of a CRAS is important since it specifies a region that enables the system to exhibit complex behaviors while ensuring safety and reaching targets. In this paper, our aim is to compute a large CRAS, i.e.\ a tight under-approximation of the maximal set of states satisfying the above condition. 


\subsection{Lyapunov-like Conditions for Characterizing CRASs}
In this subsection, we revisit Lyapunov-like conditions for characterizing CRASs and formulate challenges in applying them to actual computation of CRASs. 

 A sufficient and necessary condition, based on the relaxation of a Bellman equation, was proposed in \cite{zhao2022inner} for characterizing robust reach-avoid sets. These sets are defined as collections of initial states that guarantee the reach of target sets while ensuring safety, irrespective of disturbances. Although the system dynamics considered in this paper do not include disturbances—instead incorporating control inputs—the framework from \cite{zhao2022inner} can still be adapted to characterize CRASs.
A modified version of the condition in \cite{zhao2022inner} can be employed to derive a Lyapunov-like condition for characterizing CRASs in this paper.


\begin{proposition}
\label{prop:ra_sup}
If there exists a Lyapunov-like function \footnote{A Lyapunov-like function is a scalar function that generalizes the classical Lyapunov function concept, which is used to analyze the stability and behavior of dynamical systems.} 
$v(\bm{x}): \widehat{\mathcal{X}}\rightarrow \mathbb{R}$, which is bounded and continuous, satisfying 
\begin{equation}
    \label{eq:ra_sup}
    \begin{cases}
      \max_{\bm{u}\in \mathcal{U}}v(\bm{f}(\bm{x},\bm{u}))-\lambda v(\bm{x})\geq 0, & \forall \bm{x}\in \mathcal{X}\setminus \mathcal{T},\\ 
       v(\bm{x})\leq 0, & \forall \bm{x}\in \widehat{\mathcal{X}}\setminus \mathcal{X},
    \end{cases}
\end{equation}
where $\lambda>1$ is a user specified value and $\widehat{\mathcal{X}}$ is a set including states which system \eqref{eq:systems} starting from the safe set $\mathcal{X}$ visits within the first step, i.e., 
\begin{equation}
    \label{eq:x_hat}
    \widehat{\mathcal{X}}\supseteq \{\bm{x}\in \mathbb{R}^n\mid \bm{x}=\bm{f}(\bm{x}_0,\bm{u}),\bm{x}_0\in \mathcal{X}, \bm{u}\in \mathcal{U}\}\cup \mathcal{X},
\end{equation}
then $\{\bm{x}\in \mathcal{X}\mid v(\bm{x})>0\} \neq \emptyset$ is a CRAS.
\end{proposition}


Proposition \ref{prop:ra_sup} transforms the task of computing CRASs into the problem of identifying a Lyapunov-like function $v(\bm{x})$ that satisfies constraint \eqref{eq:ra_sup}. However, the search for a Lyapunov-like function $v(\bm{x}): \widehat{\mathcal{X}}\rightarrow \mathbb{R}$ that fulfills \eqref{eq:ra_sup} is inherently challenging. In addition to the nonlinearity, the presence of the maximum operator in the term $\max_{\bm{u}\in \mathcal{U}}v(\bm{f}(\bm{x},\bm{u}))$ further complicates the resolution of constraint \eqref{eq:ra_sup}. In order to eliminate the maximum operator, a traditional approach is to introduce a feedback control function $\bm{u}(\bm{x}): \mathbb{R}^n \rightarrow \mathcal{U}$, leading to a Lyapunov-like condition formulated in Corollary \ref{coro:ra}.
\begin{corollary}
\label{coro:ra}
If there exists a Lyapunov-like function $v(\bm{x}): \widehat{\mathcal{X}}\rightarrow \mathbb{R}$, which is bounded and continuous, and a control function $\bm{u}(\bm{x}): \mathbb{R}^n \rightarrow \mathcal{U}$, satisfying 
\begin{equation}
    \label{eq:ra_sup1}
    \begin{cases}
      v(\bm{f}(\bm{x},\bm{u}(\bm{x})))-\lambda v(\bm{x})\geq 0, & \forall \bm{x}\in \mathcal{X}\setminus \mathcal{T},\\ 
       v(\bm{x})\leq 0, & \forall \bm{x}\in \widehat{\mathcal{X}}\setminus \mathcal{X},
    \end{cases}
\end{equation}
where $\lambda>1$ is a user specified value and $\widehat{\mathcal{X}}$ is a set satisfying \eqref{eq:x_hat},  
then $\{\bm{x}\in \mathcal{X}\mid v(\bm{x})>0\} \neq \emptyset$ is a CRAS.
\end{corollary}



One notable observation is that, with a polynomial feedback control function $\bm{u}(\bm{x}):\mathbb{R}^n\rightarrow \mathcal{U}$ at hand, the quest for a polynomial Lyapunov-like function $v(\bm{x}):\widehat{\mathcal{X}}\rightarrow \mathbb{R}$ that satisfies \eqref{eq:ra_sup1} can be recast as a semi-definite program, which falls within the convex optimization framework, leveraging the sum of squares decomposition for multivariate polynomials. Nonetheless, the selection of a polynomial feedback control function that guarantees the existence of a Lyapunov-like function fulfilling \eqref{eq:ra_sup1} proves to be a formidable challenge, with no efficient techniques currently available to identify such functions. On the other hand, even if a polynomial control function that guarantees the existence of a Lyapunov-like function fulfilling \eqref{eq:ra_sup1} is chosen and a corresponding Lyapunov-like function is found, how to efficiently improve the control function and the Lyapunov-like function to obtain less conservative CRASs remains an open problem.

\begin{remark}
In the existing literature, when the system is represented as a polynomial ordinary differential equation that is affine with respect to control inputs, a practical iterative approach can be employed to estimate CRASs. This method alternately updates polynomial Lyapunov-like functions and polynomial feedback control functions using convex optimization, with an initial guess on the feedback control function $u(\bm{x}): \mathbb{R}^n\rightarrow \mathcal{U}$ or the Lyapunov-like function. However, this method is not directly applicable to the discrete-time scenario of interest in this paper. The discrete-time case is more complex and challenging. The composition of the feedback control function and Lyapunov-like function $v(\bm{f}(\bm{x},\bm{u}))$ makes the update of the feedback control function $\bm{u}(\bm{x}): \mathbb{R}^n\rightarrow \mathcal{U}$ non-convex (we may encounter nonlinear terms such as $\bm{u}^2(\bm{x})$, $\bm{u}^3(\bm{x})$, $\ldots$), when the Lyapunov function $v(\bm{x})$ is available.
\end{remark}


\section{Initial Computations of CRASs}
\label{sec:ra}
In this section, we introduce a convex optimization approach for estimating CRASs from a probabilistic standpoint. Specifically, we reformulate the estimation of CRASs as the computation of reach-avoid sets in a probabilistic context by treating the control inputs as disturbance noises following specific probability distributions. This approach results in the elimination of the variable $\bm{u}$ in the composition of control inputs and Lyapunov-like functions, achieved through the use of the expectation operator.  Consequently, convex constraints only involving the search for Lyapunov-like functions are constructed for characterizing CRASs.

We introduce randomness into control inputs by assigning it a specific probability distribution. The vectors $\bm{u}(0)$, $\bm{u}(1)$, $\ldots$, are drawn independently and identically distributed (i.i.d.) from the probability space $(\Omega,\mathcal{F},P)$, satisfying that for any measurable set $\mathcal{A}\subseteq \mathcal{U}$, $$
\text{Prob}(\bm{u}(l)\in \mathcal{A})=P(\mathcal{A}), \forall l\in \mathbb{N},
$$ with $P(\mathcal{U})=1$. The expectation under this distribution is denoted by $E[\cdot]$. Consequently, a control policy $\pi$ for system \eqref{eq:systems} can be viewed as a realization of a stochastic process $\{\bm{u}(i):\Omega\rightarrow \mathcal{U},i\in \mathbb{N}\}$ defined on the canonical sample space $\mathcal{U}^{\infty}$, with a probability measure $P^{\infty}$ induced by $P$. The expectation associated with $P^{\infty}$ is denoted by $E^{\infty}[\cdot]$. By endowing the control input with randomness in this manner, we are now in a position to formally define $p$-reach-void sets within this probabilistic framework, where $p\in [0,1)$.



\begin{definition}[$p$-reach-avoid sets, \cite{xue2021reach}]
\label{def:prs}
A $p$-reach-avoid set $\mathcal{R}_p$ is a set of initial states $\bm{x}_0$ that each gives rise to a set of trajectories which, with probability being larger than $p$, eventually enter the target set $\mathcal{T}$ while remaining inside the safe set $\mathcal{X}$ until the target is hit, i.e., 
\begin{equation*}
P^{\infty}\Big(\exists k\in \mathbb{N}. \bm{\phi}^{\pi}_{\bm{x}_0}(k)\in \mathcal{T}\bigwedge\forall l\in [0,k]\cap \mathbb{N}. \bm{\phi}^{\pi}_{\bm{x}_0}(l)\in \mathcal{X}\Big)> p.
\end{equation*} 
\end{definition}

As noted earlier in Remark 1 of \cite{xue2021reach}, there exists a non-empty set of control policies $\pi$ that can safely guide the system \eqref{eq:systems} from $\mathcal{R}_0$ to the target $\mathcal{T}$. In this paper, we further demonstrate the equivalence between a CRAS $\mathcal{R}$ and a $0$-reach-avoid set $\mathcal{R}_0$, under Assumption \ref{assum}. This relationship is elaborated upon in Theorem \ref{theo:equal}.

\begin{assumption}
\label{assum}
For any measurable set  $\mathcal{B}\subseteq \mathcal{U}$ with non-empty interior, the probability measure 
$P(\mathcal{B})>0$.
\end{assumption}

\begin{theorem}
    \label{theo:equal}
   Under Assumption \ref{assum}, a CRAS $\mathcal{R}$ is a $0$-reach-avoid set $\mathcal{R}_0$, and vice versa.
\end{theorem}
\begin{proof}
    We establish the equivalence between the CRAS $\mathcal{R}$ and the $0$-reach-avoid set $\mathcal{R}_0$ by demonstrating their mutual inclusion. Specifically, we prove that $\mathcal{R} \subseteq \mathcal{R}_0$ and $\mathcal{R}_0 \subseteq \mathcal{R}$.
    \begin{enumerate}
        \item $\bm{\mathcal{R} \subseteq \mathcal{R}_0}$. For any state $\bm{x}_0 \in \mathcal{R}$, according to the definition of $\mathcal{R}$, there exists a hitting time $k_0$ and a control policy $\pi_0=\{\bm{u}(0),\dots, \bm{u}(k_0)\}$ such that $\bm{\phi}^{\pi_0}_{\bm{x}_0}(k_0)\in\mathcal{T}$ and for all $t \in \{0,1, \ldots, k_0\}$, $\bm{\phi}^{\pi_0}_{\bm{x}_0}(t)\in\mathcal{X}$. As both $\mathcal{T}$ and $\mathcal{X}$ are open sets, there exists an $\epsilon>0$ such that $\mathcal{B}(\bm{\phi}^{\pi_0}_{\bm{x}_0}(k_0),\epsilon)\subseteq\mathcal{T}$ and for all $t \in \{1, \ldots, k_0-1\}$, $\mathcal{B}(\bm{\phi}^{\pi_0}_{\bm{x}_0}(t),\epsilon)\subseteq\mathcal{X}$. Assume that the Lipschitz constants of $\bm{f}(\bm{x}, \bm{u})$ with respect to $\bm{x}$ and $\bm{u}$ are $\mathcal{L}_{x}$ and $\mathcal{L}_{u}$ respectively. Let $\mathcal{L}=\max\{\mathcal{L}_x,\mathcal{L}_u\}$ and $\delta<\frac{\epsilon}{\mathcal{L}^{k_0}+\mathcal{L}^{k_0-1}
        +\dots+\mathcal{L}}$ 
        .For any control policy $\pi_0^\prime = \{\bm{u}^\prime(0),\dots, \bm{u}^\prime(k_0)\}$, where $\bm{u}^\prime(t)\in\mathcal{B}(\bm{u}(t),\delta)\cap\mathcal{U}$ for all $t\in\{0,1,\dots,k_0-1\}$, we have 
        \begin{equation*}
            \Vert\bm{\phi}^{\pi_0^\prime}_{\bm{x}_0}(1)-\bm{\phi}^{\pi_0}_{\bm{x}_0}(1)\Vert=\Vert \bm{f}(\bm{x}_0,\bm{u}^{\prime}(0))-\bm{f}(\bm{x}_0,\bm{u}(0))\Vert\leq\mathcal{L}_u\Vert\bm{u}^\prime(0)-\bm{u}(0)\Vert<\epsilon,
        \end{equation*}
        which means that $\bm{\phi}^{\pi_0^\prime}_{\bm{x}_0}(1)\in\mathcal{B}(\bm{\phi}^{\pi_0}_{\bm{x}_0}(1),\epsilon)\subseteq\mathcal{X}.$
        
        Similarly, we have
        \begin{equation*}
            \begin{split}
              &\Vert\bm{\phi}^{\pi_0^\prime}_{\bm{x}_0}(2)-\bm{\phi}^{\pi_0}_{\bm{x}_0}(2)\Vert\\
              =&\Vert \bm{f}(\bm{\phi}^{\pi_0^\prime}_{\bm{x}_0}(1),\bm{u}^{\prime}(1))-\bm{f}(\bm{\phi}^{\pi_0}_{\bm{x}_0}(1),\bm{u}(1))\Vert\\
              =&\Vert \bm{f}(\bm{f}(\bm{x}_0,\bm{u}^{\prime}(0)),\bm{u}^{\prime}(1))-\bm{f}(\bm{f}(\bm{x}_0,\bm{u}(0)),\bm{u}(1))\Vert\\
              =&\Vert \bm{f}(\bm{f}(\bm{x}_0,\bm{u}^{\prime}(0)),\bm{u}^{\prime}(1))-\bm{f}(\bm{f}(\bm{x}_0,\bm{u}(0)),\bm{u}^{\prime}(1))\\
              &+\bm{f}(\bm{f}(\bm{x}_0,\bm{u}(0)),\bm{u}^{\prime}(1))-\bm{f}(\bm{f}(\bm{x}_0,\bm{u}(0)),\bm{u}(1))\Vert\\
              \leq&\mathcal{L}_x\Vert \bm{f}(\bm{x}_0,\bm{u}^{\prime}(0))-\bm{f}(\bm{x}_0,\bm{u}(0))\Vert+\mathcal{L}_u\Vert\bm{u}^{\prime}(1)-\bm{u}(1)\Vert\\
              \leq&\mathcal{L}_x\mathcal{L}_u\Vert\bm{u}^{\prime}(0)-\bm{u}(0)\Vert+\mathcal{L}_u\Vert\bm{u}^{\prime}(1)-\bm{u}(1)\Vert\\
              <&\epsilon.
            \end{split}
        \end{equation*}
        Inductively, we have
        \[\forall t\in\{1,2, \ldots, k_0\},\Vert\bm{\phi}^{\pi^\prime}_{\bm{x}_0}(t)-\bm{\phi}^{\pi}_{\bm{x}_0}(t)\Vert\leq\sum_{l=1}^{t}\mathcal{L}_x^{l-1}\mathcal{L}_u\Vert\bm{u}^\prime(t-l)-\bm{u}(t-l)\Vert.\]
        Thus $\forall t\in\{1,2, \ldots, k_0\},\Vert\bm{\phi}^{\pi^\prime}_{\bm{x}_0}(t)-\bm{\phi}^{\pi}_{\bm{x}_0}(t)\Vert<\epsilon$, which means $\bm{\phi}^{\pi^\prime}_{\bm{x}_0}(k_0) \in \mathcal{T}$ and $\forall t\in\{1,2, \ldots, k_0-1\}, \bm{\phi}^{\pi_0^\prime}_{\bm{x}_0}(t)\in\mathcal{B}(\bm{\phi}^{\pi_0}_{\bm{x}_0}(t),\epsilon)\subseteq\mathcal{X}$. Then we have
        \begin{equation*}
            \begin{split}
                &P^{\infty}\Big(\exists k\in \mathbb{N}. \bm{\phi}^{\pi}_{\bm{x}_0}(k)\in \mathcal{T}\bigwedge\forall l\in [0,k]\cap \mathbb{N}. \bm{\phi}^{\pi}_{\bm{x}_0}(l)\in \mathcal{X}\Big)\\
                \geq  &P^{\infty}\Big(\forall t\in[0,k_0-1]\cap\mathbb{N}.  \bm{u}^{\prime}(t)\in\mathcal{B}(\bm{u}(t),\delta)\cap\mathcal{U}\Big)\\
                >&0.
            \end{split}
        \end{equation*}

        Consequently, $\bm{x}_0 \in \mathcal{R}_0$, and further, $\mathcal{R}\subseteq \mathcal{R}_0$.

        \item $\bm{\mathcal{R}_0 \subseteq \mathcal{R}}$. According to the definition of $0$-reach-avoid set, the probability of that starting from $\bm{x}_0$ in $\mathcal{R}_0$, there exists at least a control policy that enables the system \eqref{eq:systems} to enter the target set within a finite time and remain in the safe set before is greater than 0. Therefore, the set of control inputs that satisfy the property mentioned before cannot be an empty set, which means $\bm{x}_0\in\mathcal{R}$. Thus $\mathcal{R}_0$ is a subset of $\mathcal{R}$, i.e., $\mathcal{R}_0 \subseteq \mathcal{R}$.
    \end{enumerate}
    The proof is completed.
\end{proof}

According to Theorem \ref{theo:equal}, the task of estimating a CRAS $\mathcal{R}$ is equivalent to estimating a $0$-reach-avoid set $\mathcal{R}_0$ without adding any additional conservatism. Therefore, we can utilize any Lyapunov-like conditions for computing a $0$-reach-avoid set $\mathcal{R}_0$ in the calculation of a CRAS $\mathcal{R}$. A Lyapunov-like condition, which was proposed in \cite{xue2024finite}, is outlined in Lemma \ref{theo:ra_exp}.


\begin{lemma} [Theorem 5 in \cite{xue2024finite} in which $\beta=0$ and $N\rightarrow \infty$]
\label{theo:ra_exp}
Given $\lambda \in (1,\infty)$, if there exists a bounded function $v(\bm{x}): \widehat{\mathcal{X}}\rightarrow \mathbb{R}$ satisfying 
\begin{equation}
    \label{eq:ra_exp0}
    \begin{cases}
       E[v(\bm{f}(\bm{x},\bm{u}))]-\lambda v(\bm{x})\geq 0, & \forall \bm{x}\in \mathcal{X}\setminus \mathcal{T},\\ 
       v(\bm{x})\leq 0, & \forall \bm{x}\in \widehat{\mathcal{X}}\setminus \mathcal{X},\\
     v(\bm{x})\leq 1, & \forall \bm{x}\in \mathcal{T},\\
    \end{cases}
\end{equation}
then $\{\bm{x}\in \mathcal{X}\mid v(\bm{x})>0\} \neq \emptyset$ is a 0-reach-avoid set, 
where the set $\widehat{\mathcal{X}}$ satisfies condition \eqref{eq:x_hat}.
\end{lemma}

As per Lemma \ref{theo:ra_exp}, a CRAS can be obtained by solving the constraint \eqref{eq:ra_exp0}. In comparing the Lyapunov-like conditions \eqref{eq:ra_sup} and \eqref{eq:ra_exp0}, it is observed that \eqref{eq:ra_exp0} utilizes the expectation operator in place of the maximum operator over the control variable. This modification represents a qualitative improvement in the computation of CRASs, transforming the non-convex constraint \eqref{eq:ra_sup} on the Lyapunov-like function into a convex constraint  \eqref{eq:ra_exp0} (i.e., if $v_1(\bm{x})$ and $v_2(\bm{x})$ satisfy \eqref{eq:ra_exp0}, $\lambda_1 v_1(\bm{x})+\lambda_2 v_2(\bm{x})$ will also satisfy \eqref{eq:ra_exp0} for any $\lambda_1+\lambda_2=1$).
This conclusion also extends to the comparison between the Lyapunov-like conditions \eqref{eq:ra_sup1} and \eqref{eq:ra_exp0}. While the constraint \eqref{eq:ra_sup1} becomes convex on the Lyapunov-like function with a known feedback control function, the selection of such a control function that ensures the existence of a Lyapunov-like function satisfying \eqref{eq:ra_sup1} proves to be a significant challenge. The constraint \eqref{eq:ra_exp0}, however, does not require the selection of such control functions. 

The constraint in \eqref{eq:ra_exp0} can be further simplified to compute CRASs. However, if there exists $v(\bm{x})$ satisfying \eqref{eq:ra_exp0},  it will satisfy the simplified constraint. The simplified constraint is shown in Proposition \ref{pro_sim}. 

\begin{proposition}
\label{pro_sim}
Given $\lambda \in (1,\infty)$, if there exists a bounded function $v(\bm{x}): \widehat{\mathcal{X}}\rightarrow \mathbb{R}$ satisfying  \oomit{\textcolor{red} {satisfying $\sup_{\bm{x}\in\widehat{\mathcal{X}}}v(\bm{x}) > 0$} and that}
\begin{equation}
    \label{eq:ra_exp}
    \begin{cases}
       E[v(\bm{f}(\bm{x},\bm{u}))]-\lambda v(\bm{x})\geq 0, & \forall \bm{x}\in \mathcal{X}\setminus \mathcal{T},\\ 
       v(\bm{x})\leq 0, & \forall \bm{x}\in \widehat{\mathcal{X}}\setminus \mathcal{X},\\
    \end{cases}
\end{equation}
then $\{\bm{x}\in \mathcal{X}\mid v(\bm{x})>0\} \neq \emptyset$ is a 0-reach-avoid set, where $\widehat{\mathcal{X}}$ is a set in \eqref{eq:x_hat}.    
\end{proposition}

\begin{proof}
 Since $v(\bm{x})$ is a bounded function, there exists $M>0$ such that $|v(\bm{x})|\leq M, \forall \bm{x}\in\widehat{\mathcal{X}}$. Thus, if $v(\bm{x})$ satisfies \eqref{eq:ra_exp}, $v'(\bm{x}):=\frac{v(\bm{x})}{M}$ will satisfies \eqref{eq:ra_exp0}. According to Lemma \ref{theo:ra_exp}, the set $\{\bm{x}\in \mathcal{X}\mid v'(\bm{x})>0\}\neq \emptyset$ is a 0-reach-avoid set. Since $\{\bm{x}\in \mathcal{X}\mid v'(\bm{x})>0\}=\{\bm{x}\in \mathcal{X}\mid v(\bm{x})>0\}$, $\{\bm{x}\in \mathcal{X}\mid v(\bm{x})>0\}\neq \emptyset$ is also a 0-reach-avoid set. 
 \end{proof}
 



If $v(\bm{x})$ satisfies \eqref{eq:ra_exp0}, it will satisfy \eqref{eq:ra_exp}. In the following, we will employ \eqref{eq:ra_exp} for computing CRASs. However, although the requirement $v(\bm{x})\leq 1, \forall \bm{x}\in \mathcal{T}$ is abandoned, we can still ensure that the set $\{\bm{x}\in \mathcal{X}\mid v(\bm{x})>0\}\neq \emptyset$ is a 0-reach-avoid set. 

Assume $\widehat{\mathcal{X}}=\{\bm{x}\in \mathbb{R}^n\mid \hat{h}(\bm{x})< 0\}$, where $\hat{h}(\bm{x})\in \mathbb{R}[\bm{x}]$, and suppose the function $v(\bm{x})$ is sought within the polynomial space. By leveraging the sum of squares decomposition for multivariate polynomials, optimization problem \eqref{eq:ra_exp} can be recast as a semi-definite program, which can be solved efficiently in polynomial time using interior-point methods. The semi-definite program is formulated as follows, i.e., \eqref{eq:sos_ra_exp}. So far, we have specialized 
the nonlinear problem of computing CRASs into a problem of solving a single convex optimization.
\begin{equation}
\label{eq:sos_ra_exp}
\begin{split}
&\max \bm{c}_v \cdot \hat{\bm{w}}_v\\
&{\rm s.t. ~}
\begin{cases}
E[v(\bm{f}(\bm{x},\bm{u}))] - \lambda v (\bm{x}) + s_1(\bm{x})h(\bm{x})-s_2(\bm{x})g(\bm{x})\in \sum[\bm{x}],\\
-v(\bm{x})+s_3(\bm{x})\hat{h}(\bm{x})-s_4(\bm{x})h(\bm{x})\in \sum[\bm{x}],\\
s_1(\bm{x}), s_2(\bm{x}), s_3(\bm{x}), s_4(\bm{x}) \in \sum[\bm{x}],
\end{cases}
\end{split}
\end{equation}
where $\bm{c}_v \cdot \hat{\bm{w}}_v = \int_\mathcal{X} v(\bm{x}) d\bm{x}$, $\hat{\bm{w}}_v$ is the constant vector computed by integrating the monomials in $v(\bm{x})\in \mathbb{R}[\bm{x}]$ over the safe set $\mathcal{X}$, 
and $\bm{c}_v$ is the vector composed of unknown coefficients in $v(\bm{x}) \in \mathbb{R}[x]$.

Below, we use a simple example to illustrate how to calculate CRASs through solving the SOS \eqref{eq:sos_ra_exp}. 

\begin{example}
    \label{ex:running}
    Consider the one-dimensional discrete-time systems:
    \[x(t+1) = x(t) + 0.01(-x(t)-x^2(t)+u(t))\]
    with the safe set $\mathcal{X} = \{x \mid x^2 < 1\}$, the target set $\mathcal{T} = \{x \mid (x-0.6)^2<0.01\}$, and the control input set $\mathcal{U} = \{u \mid -1 \leq u \leq 1\}$. In this example, we set $\lambda$ to $1.01$. The degree of the polynomial $v(\bm{x})$ is set to 4, and the degrees of the polynomials $s_1(\bm{x}), \ldots, s_4(\bm{x})$ are set to 8. By solving the SOS \eqref{eq:sos_ra_exp} with YALMIP \cite{lofberg2004yalmip} and the solver Mosek \cite{aps2019mosek}, we obtain a CRAS set $\{x \mid 0.1391 < x < 0.9299\}$. The computation time for solving \eqref{eq:sos_ra_exp} is $1.78$ seconds. Utilizing Proposition \ref{pro_sim}, we efficiently compute the CRAS using the 0-reach-avoid set, thereby mitigating the nonlinear terms arising from the composition of Lyapunov-like functions and control inputs.
\end{example}

Despite the demonstration of equivalence between $0$-reach-avoid sets in the probabilistic sense and CRASs in Theorem \ref{theo:equal}, the practical computations may yield conservative CRASs using the convex optimization approach (i.e., \eqref{eq:sos_ra_exp}) for calculating $0$-reach-avoid sets to compute CRASs. The main reason for this conservativeness is the choice of the probability distribution on the control inputs. While the equivalence between $0$-reach-avoid sets and CRASs is guaranteed for any probability distributions satisfying the assumption in Theorem \ref{theo:equal}, the specific choice of distribution could significantly affect the computed CRASs, especially when using a fixed parameterized Lyapunov-like function. However, determining the optimal distribution associated with the chosen parameterized Lyapunov-like function is a challenging task, and thus, we have not placed emphasis on this issue in implementing \eqref{eq:sos_ra_exp}. Instead, we have chosen the uniform distribution throughout the numerical examples in this paper. This conservatism can also be understood through 
 the inequality $E[v(\bm{f}(\bm{x},\bm{u}))]\leq \max_{\bm{u}\in \mathcal{U}} v(\bm{f}(\bm{x},\bm{u}))$, which indicates that if a Lyapunov-like function $v(\bm{x})$ satisfies \eqref{eq:ra_exp}, it will also satisfy \eqref{eq:ra_sup}, which thus is a potentially less conservative condition. This suggests that more effective Lyapunov-like functions, leading to less conservative CRASs, may be discovered. To address this conservatism, we will present a method to refine the computed Lyapunov-like function $v(\bm{x})$ by adjusting the probability distribution of control inputs with the help of \eqref{eq:ra_sup} in the following section.


 

\section{Improving Probabilistic CRASs Estimation}
\label{sec:update}
In this section, inspired by the $\epsilon$-greedy strategy widely used in reinforcement learning, we propose an approach that iteratively updates Lyapunov-like functions to compute less conservative CRASs. 

Assume a Lyapunov-like function $v_0(\bm{x})\in \mathbb{R}[\bm{x}]$ is obtained via solving \eqref{eq:sos_ra_exp}. As mentioned before, \[\max_{\bm{u}\in \mathcal{U}}v_0(\bm{f}(\bm{x},\bm{u}))\geq E[v_0(\bm{f}(\bm{x},\bm{u}))]\] holds for $\bm{x}\in \mathcal{X}\setminus \mathcal{T}$, which implies that $v_0(\bm{x})$ is also a solution to constraint \eqref{eq:ra_sup}. This indicates that if a vector-valued control function $\bm{u}_0(\bm{x}):\mathbb{R}^n\rightarrow \mathbb{R}^m$ satisfying $\max_{\bm{u}\in \mathcal{U}}v_0(\bm{f}(\bm{x},\bm{u}))=v_0(\bm{f}(\bm{x},\bm{u}_0(\bm{x})))$ for $\bm{x}\in \mathcal{X}\setminus \mathcal{T}$ is available, solving the following constraint 
\begin{equation}
\label{eq:ra_sup10}
\begin{cases}
v(\bm{f}(\bm{x},\bm{u}_0(\bm{x})))-\lambda v(\bm{x})\geq 0, &\forall \bm{x}\in \mathcal{X}\setminus \mathcal{T},\\
v(\bm{x})\leq 0, &\forall \bm{x}\in \widehat{\mathcal{X}}\setminus \mathcal{X},
\end{cases}
\end{equation} is conducive to exploring new Lyapunov-like functions and thus may yield less conservative CRASs. However, the vector-valued control function $\bm{u}_0(\bm{x}):\mathbb{R}^n\rightarrow \mathbb{R}^m$ is fundamentally challenging to obtain. A traditional approach is to compute an approximation $\tilde{\bm{u}}_0(\bm{x})$  of the vector-valued control function $\bm{u}_0(\bm{x}):\mathbb{R}^n\rightarrow \mathbb{R}^m$, and then use this approximation $\tilde{\bm{u}}_0(\bm{x})$ to refine $v_0(\bm{x})$. Since $v_0(\bm{f}(\bm{x},\tilde{\bm{u}}_0(\bm{x})))$ may not be equal to $\max_{\bm{u}\in \mathcal{U}}v(\bm{f}(\bm{x},\bm{u}))$ for $\bm{x}\in \mathcal{X}\setminus \mathcal{T}$, we employ the 
$\epsilon$-greedy strategy to mitigate this non-optimality issue and cover the optimal one. This strategy assigns a large weight to the approximate control input (i.e., exploitation) and smaller weights to other control inputs (i.e., exploration). This balance between exploration and exploitation is crucial for the computing process, enabling it to find an optimal control input over time. In the following, we will detail our computations.



Firstly, we parameterize the control function $\tilde{\bm{u}}_0(\bm{a},\bm{x})$, which is a polynomial over $\bm{x}$ with unknown coefficients $\bm{a}$. We use this function and a family of data to identify $\bm{a}$ via solving a semi-definite program. 
Specifically, we first sample $N$ states $\{\bm{x}_i\}_{i=1}^N$ from the set $\mathcal{X}\setminus\mathcal{T}$. Since $v_0(\bm{f}(\bm{x}_i, \bm{u}))$ is nonlinear over $\bm{u}$, directly solving $\max_{\bm{u}\in \mathcal{U}}v_0(\bm{f}(\bm{x}_i, \bm{u}))$ is fundamentally challenging.  Thus, we discretize the set $\mathcal{U}$ to obtain a family of control inputs $\{\bm{u}_i\}_{i=1}^M$ and further identify the one such that  $v_0(\bm{f}(\bm{x}_i, \bm{u}))$ achieves the maximum for each sampled state $\bm{x}_i$, $i = 1, \ldots, N$, as shown in \eqref{eq:max_u_1}. Without loss of generality, assume 
\begin{equation}
\label{eq:max_u_1}
\bm{u}_i = ~ \mathop{\argmax}_{\bm{u}\in \{\bm{u}_1, \bm{u}_2, \ldots,\bm{u}_M\}} ~ v(\bm{f}(\bm{x}_i, \bm{u})). 
\end{equation}
After collecting the data $\{(\bm{x}_i,\bm{u}_i)\}_{i=1}^N$, we construct an optimization over $\bm{a}$ to fit these data while ensuring $\tilde{\bm{u}}_0(\bm{a},\bm{x}) \in \mathcal{U}$ for $\bm{x}\in \mathcal{X}\setminus \mathcal{T}$.  The optimization, which can be encoded as a semi-definite program using the SOS decomposition for multivariate polynomials as detailed in \eqref{sos:u} in Appendix \ref{app:sos}, is presented below. 
\begin{equation}
\begin{split}
&\min_{\bm{a}} \sum_{i=1,\ldots,N} \|\tilde{\bm{u}}_0(\bm{a},\bm{x}_i) - \bm{u}_i \|_2^2,\\
{~~ \rm s.t. ~}&
\tilde{\bm{u}}_0(\bm{a},\bm{x}) \in \widehat{\mathcal{U}}, ~\forall \bm{x}\in\mathcal{X}\setminus\mathcal{T},
\end{split}
\label{eq:updata_u_2}
\end{equation}
where $\widehat{\mathcal{U}} = \{\bm{u}\mid \underline{\bm{u}} + \delta \leq \bm{u} \leq \overline{\bm{u}} - \delta\}$ is a subset of the set $\mathcal{U}$, and $\delta$ is a user specified value satisfying $0 < \delta < \min_{1\leq j\leq m}\frac{1}{2}(\underline{\bm{u}}(j)-\overline{\bm{u}}(j))$. The reason that a positive value $\delta$ is introduced to construct such a set $\widehat{\mathcal{U}}$ will be explained later.

Assume $\bm{a}_0$ is the solution obtained by solving \eqref{eq:updata_u_2}. Given that $\tilde{\bm{u}}_0(\bm{a}_0,\bm{x})$ is an approximation of $\bm{u}_0(\bm{x})=\arg\max_{\bm{u}\in \mathcal{U}} v_0(\bm{f}(\bm{x},\bm{u}))$ for $\bm{x}\in \mathcal{X}\setminus \mathcal{T}$, traditional approaches suggest using $\tilde{\bm{u}}_0(\bm{a}_0,\bm{x})$ to update the Lyapunov function $\bm{v}_0(\bm{x})$ by solving \eqref{eq:ra_sup1}. However, this method may not effectively refine the Lyapunov-like function, especially when $\tilde{\bm{u}}_0(\bm{a}_0,\bm{x})$ is not a good approximation, as it does not select the optimal control input for computations. To overcome this shortcoming, we adopt the $\epsilon$-greedy strategy, a widely used technique in reinforcement learning, which assigns a high probability to selecting control inputs close to $\tilde{\bm{u}}_0(\bm{a}_0,\bm{x})$ and a low probability to others. This strategy not only encourages the use of $\tilde{\bm{u}}_0(\bm{a}_0,\bm{x})$, but also considers all control inputs which cover the optimal one in $\mathcal{U}$, thereby facilitating the refinement of the computed Lyapunov-like function. Building on these insights, we introduce a new constraint for updating the Lyapunov-like function, which is detailed in Proposition \ref{new_constraint}. This constraint can be encoded as an SOS program, as illustrated in \eqref{sos:v_2} in Appendix \ref{app:sos}. For the sake of simplicity, we shall denote $\tilde{\bm{u}}_{0}(\bm{x})$ as an abbreviation for $\tilde{\bm{u}}_{0}(\bm{a}_0,\bm{x})$ unless otherwise specified.
\begin{proposition}
\label{new_constraint}
Given positive values $\epsilon\in [0,1]$, $\delta \in (0, \min_{1\leq j\leq m}\frac{1}{2}(\underline{\bm{u}}(j)-\overline{\bm{u}}(j)))$, and $\lambda \in (1,\infty)$, if there exists a function $v(\bm{x})\in \mathbb{R}[\bm{x}]$ such that
\begin{equation}
\label{eq:v_2}
\begin{cases}
\frac{1-\epsilon}{Z}\int_{\tilde{\bm{u}}_0(\bm{x})-\delta}^{\tilde{\bm{u}}_0(\bm{x})+\delta} v(\bm{f}(\bm{x},\bm{u})) d\bm{u} + \frac{\epsilon}{Z}\int_{\underline{\bm{u}}}^{\overline{\bm{u}}} v(\bm{f}(\bm{x},\bm{u})) d\bm{u} \geq \lambda v(\bm{x}),&\forall \bm{x} \in \mathcal{X}\setminus\mathcal{T},\\ 
    v(\bm{x}) \leq 0, &\forall \bm{x} \in \widehat{\mathcal{X}}\setminus\mathcal{X},
    \end{cases}
    \end{equation}
    where $Z = (1-\epsilon)(2\delta)^m + \epsilon \prod_{j=1}^m (\overline{\bm{u}}(j) - \underline{\bm{u}}(j))$, 
    then the set $\{\bm{x}\in \mathcal{X}\mid v(\bm{x})>0\}$ is a CRAS.
    \label{coro:2}
\end{proposition}
\begin{proof}
    For any $\epsilon \in [0, 1]$ and $\delta \in (0, \min_{1\leq j\leq m}\frac{1}{2}(\underline{\bm{u}}(j)-\overline{\bm{u}}(j)))$, we have 
    \[\frac{1-\epsilon}{Z}\int_{\tilde{\bm{u}}_0(\bm{x})-\delta}^{\tilde{\bm{u}}_0(\bm{x})+\delta}  d\bm{u} + \frac{\epsilon}{Z}\int_{\underline{\bm{u}}}^{\overline{\bm{u}}} d\bm{u}=1\]
    and thus, 
    \begin{equation*}
        \max_{\bm{u}\in \mathcal{U}}v(\bm{f}(\bm{x},\bm{u})) \geq \frac{1-\epsilon}{Z}\int_{\tilde{\bm{u}}_0(\bm{x})-\delta}^{\tilde{\bm{u}}_0(\bm{x})+\delta} v(\bm{f}(\bm{x},\bm{u})) d\bm{u} + \frac{\epsilon}{Z}\int_{\underline{\bm{u}}}^{\overline{\bm{u}}} v(\bm{f}(\bm{x},\bm{u})) d\bm{u}
    \end{equation*}
    Thus, if there exists $v(\bm{x}) \in \mathbb{R}[\bm{x}]$ satisfying constraint \eqref{eq:v_2}, then $v(\bm{x})$ must satisfy constraint \eqref{eq:ra_sup}. According to Proposition \ref{prop:ra_sup}, the conclusion holds.
\end{proof}

Similar to the constraint \eqref{eq:ra_exp}, the problem of solving \eqref{new_constraint} can also be relaxed into a semi-definite program, which is shown in \eqref{sos:v_2} in Appendix, via the SOS decomposition for multivariate polynomials. 

\begin{example}
    \label{ex:running_2}
    To illustrate the procedure for enlarging the CRAS computed via SOS optimization (Eq. \eqref{eq:sos_ra_exp}), we employ the system described in Example \ref{ex:running}. Within Example \ref{ex:running}, an initial CRAS $\{x \mid 0.1391 < x < 0.9299\}$ was computed. 
    
    We set $\delta = 0.1$ and $\epsilon=0.5$.  Then, we sample $50$ states from $\mathcal{X} \setminus \mathcal{T}$, and discretize the control set $\mathcal{U}$ into $5$ candidate inputs. For each sampled state, the optimal control input is obtained by solving \eqref{eq:max_u_1}. Subsequently, a 6-degree polynomial controller is fitted using \eqref{eq:updata_u_2} via solving the semidefinite program \eqref{sos:u} (visualized in Fig. \ref{fig:running}). Since the controller is constrained to the set $\widehat{\mathcal{U}}=[-0.9,0.9]$, its $\delta$-neighborhood (represented by the blue region in Fig. \ref{fig:running}) remains within $U$, ensuring the validity of the resulting CRAS. Finally, via solving the constraint \eqref{eq:ra_exp} with $\delta=0.1$ and $\epsilon=0.5$ using the sem-definite program \eqref{sos:v_2}, we obtain an expanded CRAS $\{x \mid -0.0290 < x < 0.9480\}$, which is larger than the initial one.

    \begin{figure}[ht]
    \centering
    \includegraphics[width=0.95\linewidth]{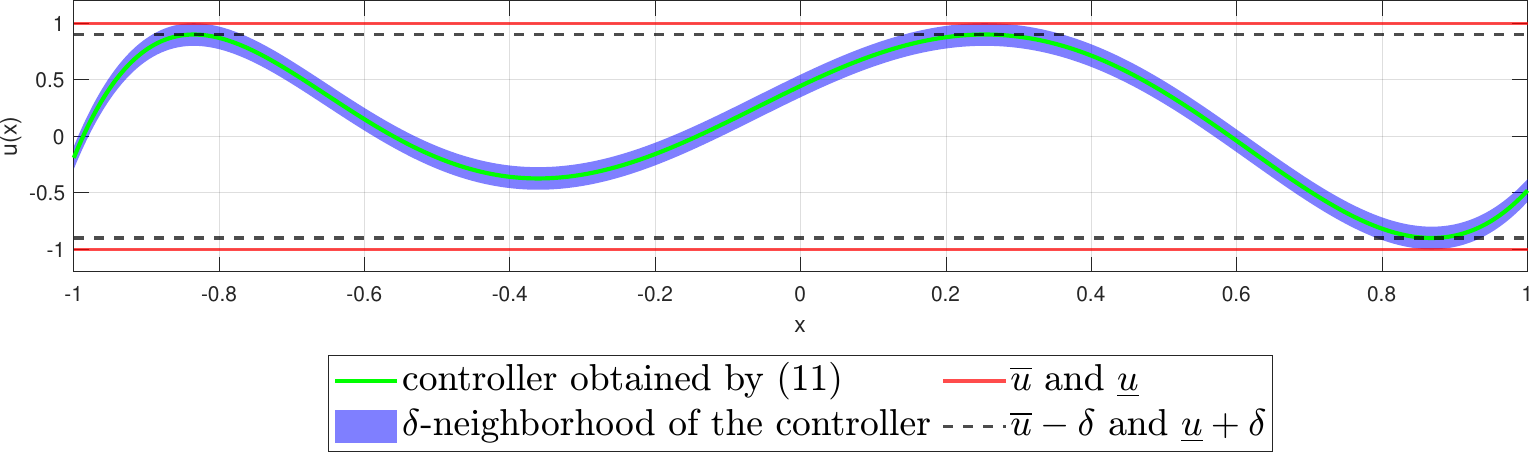}
    \caption{The controller obtained by solving \eqref{eq:updata_u_2}.}
    \label{fig:running}
    \end{figure}
\end{example}

The expression $\frac{1-\epsilon}{Z}\int_{\tilde{\bm{u}}_0(\bm{x})-\delta}^{\tilde{\bm{u}}_0(\bm{x})+\delta} v(\bm{f}(\bm{x},\bm{u})) d\bm{u} + \frac{\epsilon}{Z}\int_{\underline{\bm{u}}}^{\overline{\bm{u}}} v(\bm{f}(\bm{x},\bm{u})) d\bm{u}$ within constraint \eqref{eq:v_2} can be interpreted as the expected value of $v(\bm{f}(\bm{x},\bm{u}))$ with respect to a probability distribution characterized by the following probability density:
\begin{equation}
    \rho(\bm{u})=
    \begin{cases}
    \frac{1}{Z}, &\text{~if~} \bm{u}\in [\tilde{\bm{u}}_0(\bm{x})-\delta, \tilde{\bm{u}}_0(\bm{x})+\delta] \subseteq \mathcal{U},\\
    \frac{\epsilon}{Z}, & \text{~if~} \bm{u}\in \mathcal{U}\setminus [\tilde{\bm{u}}_0(\bm{x})-\delta, \tilde{\bm{u}}_0(\bm{x})+\delta].
    \end{cases}
\end{equation}
We can guarantee $[\tilde{\bm{u}}_0(\bm{x})-\delta, \tilde{\bm{u}}_0(\bm{x})+\delta] \subseteq \mathcal{U}$, since we require  $\tilde{\bm{u}}_0(\bm{x}) \in \widehat{\mathcal{U}}, ~\forall \bm{x}\in\mathcal{X}\setminus\mathcal{T}$ holds in \eqref{eq:updata_u_2}. \textit{This is the reason that a positive value $\delta>0$ is introduced in the optimization} \eqref{eq:updata_u_2}. In contrast to constraint \eqref{eq:ra_exp}, the probability density function in constraint \eqref{eq:v_2} becomes state-dependent rather than state-independent. This change poses an open question as to whether the computed CRAS $\{\bm{x} \in \mathcal{X} \mid v(\bm{x}) > 0\}$ qualifies as a 0-reach-avoid set.


When $\epsilon=1$, \eqref{eq:v_2} is equivalent to \eqref{eq:ra_exp} under the assumption of a uniform distribution imposed over the control set $\mathcal{U}$, as the expression $$
\frac{1-\epsilon}{Z}\int_{\tilde{\bm{u}}_0(\bm{x})-\delta}^{\tilde{\bm{u}}_0(\bm{x})+\delta} v(\bm{f}(\bm{x},\bm{u})) d\bm{u} + \frac{\epsilon}{Z}\int_{\underline{\bm{u}}}^{\overline{\bm{u}}} v(\bm{f}(\bm{x},\bm{u})) d\bm{u}
$$ reduces to $
\frac{1}{Z'}\int_{\underline{\bm{u}}}^{\overline{\bm{u}}} v(\bm{f}(\bm{x},\bm{u})) d\bm{u}$, where $Z'=\prod_{j=1}^m (\overline{\bm{u}}(j) - \underline{\bm{u}}(j))$. Conversely, if $\epsilon=0$, the expression becomes $$
\frac{1-\epsilon}{Z}\int_{\tilde{\bm{u}}_0(\bm{x})-\delta}^{\tilde{\bm{u}}_0(\bm{x})+\delta} v(\bm{f}(\bm{x},\bm{u})) d\bm{u} + \frac{\epsilon}{Z}\int_{\underline{\bm{u}}}^{\overline{\bm{u}}} v(\bm{f}(\bm{x},\bm{u})) d\bm{u}
$$ which simplifies to 
\begin{equation}
\label{nei_u0}
\frac{1}{Z^{''}}\int_{\tilde{\bm{u}}_0(\bm{x})-\delta}^{\tilde{\bm{u}}_0(\bm{x})+\delta} v(\bm{f}(\bm{x},\bm{u})) d\bm{u},
\end{equation}
 with $Z^{''}=(2\delta)^m$. This indicates that only control inputs within a $\delta$-neighborhood of $\tilde{\bm{u}}_0(\bm{x})$ are considered for computation. This approach is akin to traditional methods that use $\tilde{\bm{u}}_0(\bm{a}_0,\bm{x})$ to update the Lyapunov function $\bm{v}_0(\bm{x})$ by solving \eqref{eq:ra_sup1}, with a minor adjustment. However, this slight modification, which involves considering control inputs within a $\delta$-neighborhood of $\tilde{\bm{u}}_0(\bm{x})$ rather than solely $\tilde{\bm{u}}_0(\bm{x})$, could offer advantages for computing less conservative CRASs, particularly when $\tilde{\bm{u}}_0(\bm{x})$ is not a suitable approximation. Besides, since $Z$ appears in the denominator in the constraint, $$
\frac{1-\epsilon}{Z}\int_{\tilde{\bm{u}}_0(\bm{x})-\delta}^{\tilde{\bm{u}}_0(\bm{x})+\delta} v(\bm{f}(\bm{x},\bm{u})) d\bm{u} + \frac{\epsilon}{Z}\int_{\underline{\bm{u}}}^{\overline{\bm{u}}} v(\bm{f}(\bm{x},\bm{u})) d\bm{u} - \lambda v(\bm{x})\geq 0, \forall \bm{x}\in \mathcal{X}\setminus \mathcal{T},
$$ we cannot take $\epsilon=0$ and $\delta=0$ simultaneously.  When $\epsilon=0$ and $\delta=0$, traditional methods that use $\tilde{\bm{u}}_0(\bm{a}_0,\bm{x})$ to update the Lyapunov function $\bm{v}_0(\bm{x})$ can be applied by solving \eqref{eq:ra_sup1}.

 

The above procedure, which updates the Lyapunov-like function $v_0(\bm{x})$ by adjusting the probability distribution on control variables, can be iterated to compute less conservative CRASs. This iterative procedure is outlined in Alg. \ref{alg}.
\begin{algorithm}[h]
\caption{An Iterative Procedure for Estimating CRASs}
\label{alg}
\begin{algorithmic}[1]
\REQUIRE an initial $\epsilon \in [0, 1)$, $\delta \in (0, \min_{1\leq j\leq m}\frac{1}{2}(\underline{\bm{u}}(j)-\overline{\bm{u}}(j)))$, $\lambda > 1$, maximum iteration numbers $K$, the degree of each polynomial used in SOS, system \eqref{eq:systems}, safe set $\mathcal{X}$, target set $\mathcal{T}$, control input set $\mathcal{U}$.
\ENSURE a CRAS $\{\bm{x}\in \mathcal{X} \mid \bigvee_{k=0}^K v_k(\bm{x}) > 0\}$.
\STATE Solve SOS \eqref{eq:sos_ra_exp} to obtain an initial CRAS $\{\bm{x}\in \mathcal{X} \mid v_0(\bm{x}) > 0\}$;
\FOR{$k\gets 1,\ldots, K$}
\STATE yield the family of data using optimization \eqref{eq:max_u_1} based on $v_{k-1}(\bm{x})$;
\STATE update the controller $\bm{u}_k(\bm{x})$ by solving SOS \eqref{sos:u};
\STATE solve SOS \eqref{sos:v_2} to obtain CRAS $\{\bm{x}\in \mathcal{X} \mid v_k(\bm{x}) > 0\}$ using the controller $\bm{u}_k(\bm{x})$;
\STATE update $\epsilon$;
\ENDFOR
\STATE \textbf{return} $\{\bm{x}\in \mathcal{X} \mid \bigvee_{k=0}^K v_i(\bm{x}) > 0\}$.
\end{algorithmic}
\end{algorithm}

\begin{remark}
In reinforcement learning, the $\epsilon$-greedy strategy, particularly with a decreasing $\epsilon$ throughout the iterative process, has been shown to achieve both an enlarged radius of convergence and an improved convergence rate \cite{zhang2024convergence}. Furthermore, a gradually decreasing $\epsilon$ often leads to superior performance compared to maintaining a fixed value. Building on this principle, our algorithm systematically reduces $\epsilon$ to effectively expand CRASs.
\end{remark}
\begin{remark}
Our framework is highly extensible and is not restricted to reach-avoid problems and constraint \eqref{eq:ra_exp}. It can be applied to compute various Lyapunov-like functions such as control barrier certificates (CBCs) for safety verification and control Lyapunov functions for stability analysis in discrete-time systems. In Appendix \ref{app:cbf} and Subsection \ref{sub:com}, we illustrate how our framework can be extended to address CBCs for safety verification in discrete-time systems and showcase its effectiveness through experiments.
\end{remark}

\begin{example}
    \label{ex:running_3}
    Consider again the system in Example \ref{ex:running}. We initialize $\epsilon_0 = 0.5$ and iteratively update it according to the rule $\epsilon_{t+1} = 0.8\epsilon_{t}$. The parameters $\delta$ and $K$ are set to $0.1$ and $5$, respectively. Applying Algorithm \ref{alg} results in a final CRAS of $\{x \mid -0.7157 < x < 0.9480\}$. Compared to the initial CRAS obtained in Example \ref{ex:running}, the volume of the CRAS increases by approximately $110\%$. The total computation time for Algorithm \ref{alg} is $8.75$ seconds.
\end{example}

\section{Experiments}
\label{sec:exp}

In this section, we first demonstrate the effectiveness of Alg. \ref{alg} through a series of experiments, and conduct a comprehensive discussion on the robustness of Alg. \ref{alg}. Afterwards, to evaluate the performance and scalability of our work, we extend the semi-definite program \eqref{eq:sos_ra_exp} to safety verification for discrete-time systems via the search of control barrier functions and compare the results with two state-of-the-art methods, Fossil \cite{edwards2024fossil} and DeepICBC \cite{ren2024formal}. Additional details about the experimental setups and results can be found in Appendix \ref{app:exp}.

In the experiment, we use Monte Carlo sampling to estimate the volume of the calculated CRAS. We sample $10^6$ points uniformly within the safe set and count how many of those points, denoted as $N$, are inside the CRAS. The ratio $\gamma = \frac{N}{10^6}$ serves as an estimate of the volume of the CRAS. For this purpose, we normalize the volume of the safe set $\mathcal{X}$ to one.

\subsection{Computation of CRASs}
In this subsection, we demonstrate the effectiveness of Algorithm \ref{alg} through eight benchmark problems. To the best of our knowledge, no prior work has been specifically designed for calculating the CRAS for deterministic discrete-time systems. To highlight the advantages of the $\epsilon$-greedy strategy, we establish a baseline using a special case of this strategy, where $\epsilon = \delta = 0$. In this scenario, the fitted controller is applied directly, and the CRAS is iteratively solved based on Corollary \ref{coro:ra}. We refer to this case as the greedy strategy.

\begin{example}
    \label{exp:ra_2d_vanderpol} Consider the VanderPol oscillator in \cite{henrion2013convex},
    \begin{equation*}
        \label{eq:exp3}
        \begin{cases}
        x(t+1)=x(t) + 0.01(-2y(t)),\\
        y(t+1)=y(t) + 0.01(0.8x(t)-10(y(t)-0.21)y(t)+u(t)),
        \end{cases}
    \end{equation*}
    with the safe set $\mathcal{X}=\{\,(x,y)^{\top}\mid x^2+y^2-1 < 0\,\}$, the target set $\mathcal{T}=\{\,(x,y)^{\top}\mid x^2+y^2-0.01< 0\,\}$, and the control input set $\mathcal{U} = \{u\mid -3 \leq u \leq 3\}$. 
\end{example}
    
In this example, we initialize $\epsilon_0 = 0.3$ and iteratively update it using the formula $\epsilon_{t+1} = 0.5\epsilon_{t}$. We set the parameter $\delta$ to $0.25$ and utilize SOS \eqref{eq:sos_ra_exp} to obtain an initial CRAS, which is estimated to have a volume of $0.5428$. Based on this initial CRAS, we apply the iterative CRAS expansion strategy outlined in Alg. \ref{alg}, which results in a significant expansion of the CRAS within $10$ iterations, ultimately achieving a volume of approximately $0.9070$. The performance of our method is visualized in Fig. \ref{fig:ra_1}. For comparison, we also employed the greedy strategy under the same parameter settings, yielding a CRAS with a volume of $0.7772$ (as shown in Fig. \ref{fig:ra_1_u}). These results demonstrate that our approach  (Alg. \ref{alg}) with $\epsilon$-greedy strategy can effectively expand the CRAS, outperforming the greedy strategy in this example. 
\begin{figure}[t]
    \centering
    \subfigure[Example \ref{exp:ra_2d_vanderpol}]{
    \includegraphics[height=4.5cm,width=4.5cm]{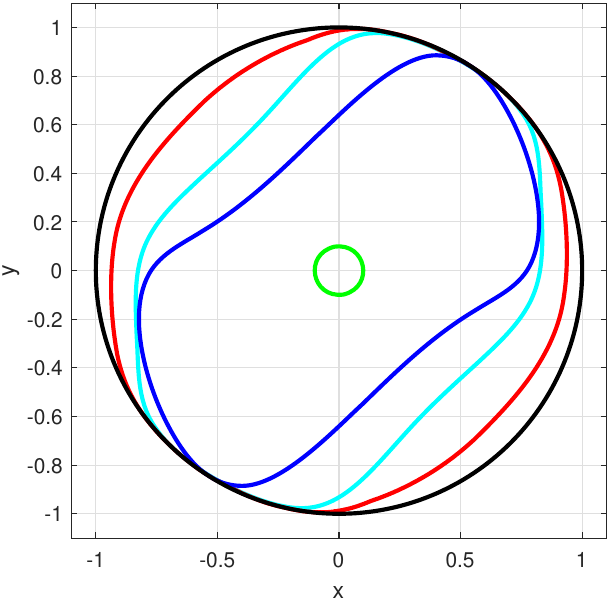}
    \label{fig:ra_1}
    }
    \subfigure[Example \ref{exp:ra_2d_vanderpol}]{
    \includegraphics[height=4.5cm,width=4.5cm]{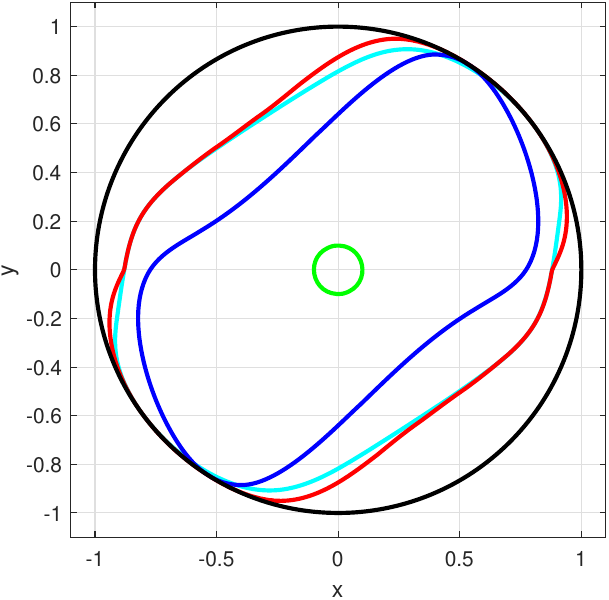}
    \label{fig:ra_1_u}
    }\\
    \subfigure[Example \ref{exp:ra_2d_2021}]{
    \includegraphics[height=4.5cm,width=4.5cm]{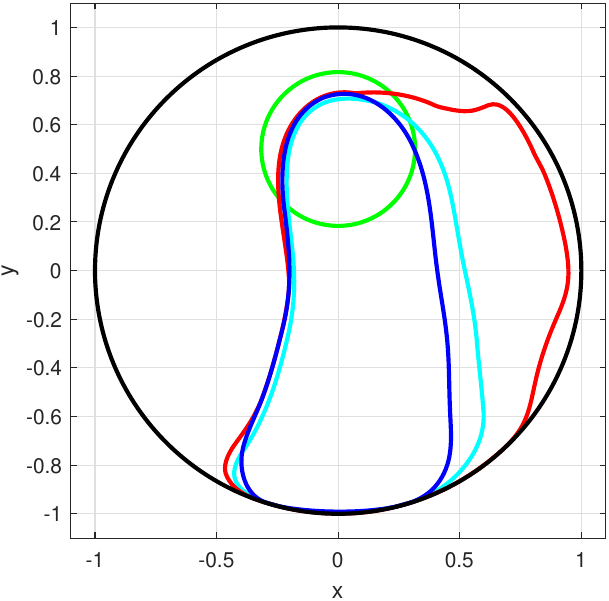}
    \label{fig:ra_2}
    }
    \subfigure[Example \ref{exp:ra_3}]{
    \includegraphics[height=4.5cm,width=4.5cm]{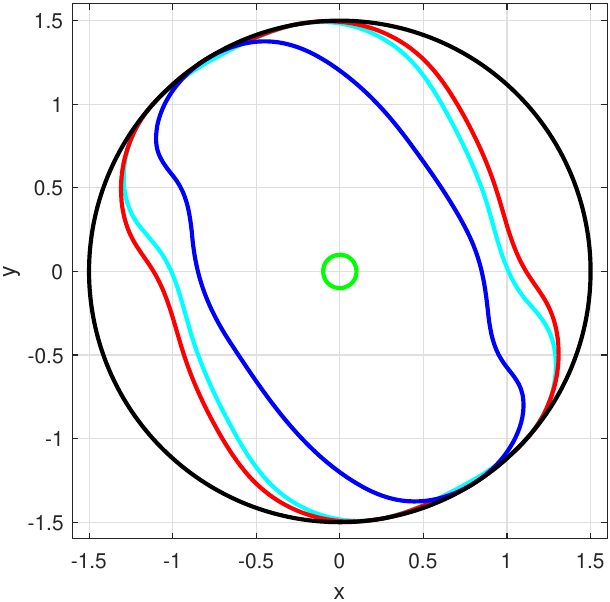}
    \label{fig:ra_3}
    }
    \caption{Black and \textcolor{green}{green} curves represent the boundaries of $\mathcal{X}$ and $\mathcal{T}$, respectively. \textcolor{blue}{Blue} and \textcolor{red}{red} curves illustrate the CRAS calculated using SOS \eqref{eq:sos_ra_exp} and Alg. \ref{alg}, respectively. In Fig. \ref{fig:ra_1} and \ref{fig:ra_1_u}, \textcolor{cyan}{cyan} curve shows the CRAS boundary for $k=3$ in Alg. \ref{alg}, and in Fig. \ref{fig:ra_2}, it shows the boundary for $k=5$.}
\end{figure}

\begin{figure}[t]
    \centering
    \includegraphics[height=3.0cm,width=11cm]{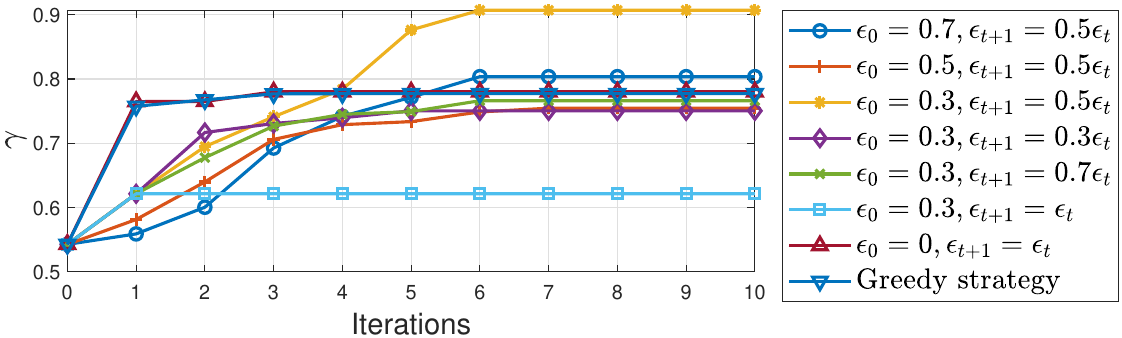}
    \caption{Results of Alg. \ref{alg} under different $\epsilon$ settings in Example \ref{exp:ra_2d_vanderpol}, where $\gamma$ denotes an estimate of the CRAS volume via random sampling.}
    \label{fig:epsilon}
\end{figure}

Additionally, we evaluate the impact of different $\epsilon$ settings on Alg. \ref{alg}, as summarized in Fig. \ref{fig:epsilon}. Under different initial values and update strategies of $\epsilon$, the convergence rates and resulting values of $\gamma$ vary. Nevertheless, all configurations successfully expand the CRAS. Notably, when $\epsilon$ is kept constant at $0.3$, the final result is relatively conservative, highlighting the necessity of adjusting $\epsilon$ during the iterative process. Furthermore, when $\epsilon$ is set to $0$ (considering only the neighborhood of the fitted controller), the final volume is $0.7801$, which is slightly better than the result obtained by the greedy strategy. Although $\epsilon=0, \delta=0.25$ leads to faster convergence, the final CRAS is smaller than the CRAS obtained when $\epsilon=0.3, \delta=0.25$, which highlights the advantage of the $\epsilon$-greedy strategy.

\begin{example}
\label{exp:ra_2d_2021}
    Consider the discrete-time model in \cite{xue2021reach},
    \begin{equation}
    \label{exp:1}
    \begin{cases}
    x(t+1)=x(t) + 0.01(-0.5x(t)-0.5y(t)+0.5x(t)y(t)),\\
    y(t+1)=y(t) + 0.01(-0.5y(t)+1+u(t)),
    \end{cases}
    \end{equation}
    with the safe set $\mathcal{X}=\{\,(x,y)^{\top}\mid x^2+y^2-1 < 0\,\}$, the target set $\mathcal{T}=\{\,(x,y)^{\top}\mid x^2+(y-0.5)^2-0.1< 0\,\}$, and the control input set $\mathcal{U}=\{\,u\mid -1 \leq u \leq 1\,\}$.
\end{example}

    In this example, we initialize $\epsilon_0=0.3$ and update it using the rule $\epsilon_{t+1} = 0.5\epsilon_{t}$. The parameter $\delta$ is set to $0.01$. Based on SOS \eqref{eq:sos_ra_exp}, we obtain a CRASs with a volume of $0.3509$. Afterwards, we consider two different scenarios for fitting the controller. In the first scenario, we utilize a second-degree polynomial as the template, sampling the state on a $10 \times 10$ grid and discretizing $10$ control inputs within the control set. In the second scenario, we employ a fourth-degree polynomial as the template, with the state sampled on a $20 \times 20$ grid and $20$ control inputs discretized within the control set. In the first scenario, Alg.  \ref{alg} produces a CRAS with a volume of $0.5716$, as illustrated in Fig. \ref{fig:ra_2}, while the greedy strategy achieves a lower CRAS volume of $0.4330$. This discrepancy primarily arises from the limited sampling data and the lower degree of the polynomial template, which diminishes the quality of the fitted controller, thereby rendering the epsilon-greedy strategy more effective. Notably, when $\epsilon_0=0$ and $\delta = 0.01$, the CRAS volume reaches $0.5075$, surpassing the result obtained using the greedy strategy; in the second scenario, a better controller can be fitted, resulting in CRAS volumes of $0.6080$ and $0.6005$ using the greedy strategy and Alg. \ref{alg}, respectively, narrowing the performance gap. This indicates that the greedy strategy is more sensitive to the quality of the fitted controller, whereas Alg. \ref{alg} is only moderately impacted. While employing higher-order polynomials and more extensive sampling generally results in improved outcomes, it also leads to increased computational time and can complicate the optimization problem. In the second scenario, Alg. \ref{alg} takes 684 seconds to execute, in contrast to just 117 seconds in the first scenario. Additionally, in this example, the MOSEK solver fails when a higher-order polynomial is used as controller templates.

\begin{example}
    \label{exp:ra_3} Consider the system adapted from \cite{tan2008stability},
    \begin{equation*}
        \label{eq:exp4}
        \begin{cases}
        x(t+1)=x(t) + 0.01(y(t) + x(t)u_1(t)),\\
        y(t+1)=y(t) + 0.01(-(1-x^2(t))x(t) - y(t) + y(t)u_2(t)),
        \end{cases}
    \end{equation*}
    with the safe set $\mathcal{X}=\{\,(x,y)^{\top}\mid x^2+y^2-1.5^2 < 0\,\}$, the target set $\mathcal{T}=\{\,(x,y)^{\top}\mid x^2+y^2-0.01< 0\,\}$, and the control set $\mathcal{U} = \{(u_1,u_2)^\top\mid -0.5 \leq u_1, u_2 \leq 0.5\}$.
\end{example}
    In this example, we begin by initializing $\epsilon_0 = 0.3$ and updating it using the rule $\epsilon_{t+1} = 0.5 \epsilon_{t}$. The parameter $\delta$ is set to $0.1$. The CRAS produced by Alg. \ref{alg} is depicted in Fig. \ref{fig:ra_3}. The initial volume of the CRAS, calculated using SOS \eqref{eq:sos_ra_exp}, is $0.5344$. After applying Algorithm \ref{alg}, the volume of the CRAS expands to $0.7887$, resulting in a significant reduction in conservatism. In contrast, the greedy strategy is unable to achieve any expansion of the CRAS. The computation times for SOS \eqref{eq:sos_ra_exp}, Algorithm \ref{alg}, and the greedy strategy are $4.56$ seconds, $132$ seconds, and $127$ seconds, respectively.

 Due to space limitations, Examples \ref{exp:ra_2d_2008}--\ref{exp:ra_6d} are provided in Appendix \ref{app:exp}. These experimental results are summarized in Table  \ref{table}, where $n$ is the state dimension, $d_v$ and $d_u$ denote the degree of the polynomial used to parameterize $v(\bm{x})$ and controllers. The time unit in the table is seconds. Our method successfully computes the CRAS for all examples. As shown in Table  \ref{table}, SOS \eqref{eq:sos_ra_exp} effectively solves the initial CRAS, and Alg. \ref{alg} significantly expands it. Compared with the greedy strategy, Alg. \ref{alg} generates less conservative CRAS in most cases, highlighting the advantages of the $\epsilon$-greedy strategy.
\begin{table}[t]
\centering
\caption{Evaluation on CRASs for Examples \ref{exp:ra_2d_2008}-\ref{exp:ra_6d}, with $\gamma$ serving as an estimate of the CRAS volume obtained through random sampling.}
\setlength{\tabcolsep}{2.3mm}{
    \begin{tabular}{*{10}{c}}
\toprule
\multirow{2}{*}{Example} & \multirow{2}{*}{$n$} & \multirow{2}{*}{$d_v$} & \multirow{2}{*}{$d_u$} & \multicolumn{2}{c}{SOS \eqref{eq:sos_ra_exp}} & \multicolumn{2}{c}{Greedy Strategy} & \multicolumn{2}{c}{Alg. \ref{alg}} \\ 
\cmidrule(lr){5-6}\cmidrule(lr){7-8}\cmidrule(lr){9-10}
& & & & $\gamma$ & Time & $\gamma$ & Time & $\gamma$ & Time \\ \midrule
\ref{exp:ra_2d_2008} & 2 & 8 & 2 & $0.6951$ & $4.29$  & $0.7634$ & $225$ & $0.8005$ & $111$ \\
\ref{exp:ra_2d_2000} & 2 & 8 & 4 & $0.8380$ & $3.82$  & $0.8380$ & $43.5$ & $0.8971$ & $52.9$ \\
\ref{exp:ra_3d_vanderpol} & 3 & 8 & 2 & $0.4247$ & $188$  & $0.6825$ & $680$ & $0.6893$ & $765$ \\ 
\ref{exp:ra_4d} & 4 & 4 & 2 & $0.8386$ & $4.76$  & $0.9044$ & $51.4$ & $0.9999$ & $48.6$ \\ 
\ref{exp:ra_6d} & 6 & 4 & 2 & $0.6575$ & $157$  & $0.7231$ & $2448$ & $0.7364$ & $2402$ \\ 
\bottomrule
\end{tabular}}
\label{table}
 
\end{table}
Our method heavily relies on solving SOS programming problems, which are exponential in the size of the state variables and the degree of the polynomials. Consequently, our method also encounters challenges in solving large-scale problems. This is because the number of constraints required to express the problem as an SDP grows rapidly, leading to an explosion in the number of variables and constraints. This issue is often referred to as the ``curse of dimensionality''. As a result, our method faces significant challenges when tackling large-scale problems. Nonetheless, despite this exponential complexity, our approach has achieved remarkable success in computing CRASs for a variety of systems. Furthermore, it outperforms existing methods in terms of efficiency, as will be demonstrated in the following subsection.

\subsection{Comparisons}
\label{sub:com}
According to our knowledge, there are no existing methods specifically tailored for computing CRASs for discrete-time systems. In contrast, there are methods for finding Lyapunov-like functions, such as control barrier functions, which are a common tool for safety verification. Our method can also be adapted for finding control barrier functions. Further details on control barrier functions and associated safety analysis are provided in Appendix \ref{app:cbf}. Therefore, in this section, we evaluate our approach against state-of-the-art techniques, including those implemented in Fossil \cite{edwards2024fossil} and DeepICBC \cite{ren2024formal}, by computing control barrier functions for safety verification. Moreover, to provide a fair comparison, we will apply our method to synthesizing control barrier functions in the following discussion,  based on all polynomial system benchmarks tested in tools Fossil and DeepICBC. Specifically, Example \ref{exp:safe1} comes from Fossil, while Examples \ref{exp:safe2}-\ref{exp:safe4} comes from DeepICBC. To further assess the scalability of our method for high-dimensional systems, we additionally test three more examples, denoted as Example \ref{exp:safe5}-\ref{exp:safe_lorenz12}. Each experiment is subject to a maximum time limit of 8 hours. The experiment results are summarized in Table  \ref{tab:safe}.
\begin{table}[t]
\centering
\caption{Safety Verification Evaluation (The dash `-' indicates verification failures)}
\setlength{\tabcolsep}{3mm}{
    \begin{tabular}{*{5}{c}}
\toprule
\multirow{2}{*}{Example} & \multirow{2}{*}{$n$}  &
\multicolumn{3}{c}{Time (seconds)}\\
\cmidrule(lr){3-5}
& & Fossil & DeepICBC & SOS \eqref{eq:sos_ra_exp_s}  \\ \midrule
\ref{exp:safe1} & 2 & $0.72$ & $76$     & $1.9$  \\
\ref{exp:safe2} & 2 & $-$    & $373$   & $1.9$ \\
\ref{exp:safe3} & 3 & $-$    & $28452$ & $24.6$ \\
\ref{exp:safe4} & 4 & $-$    & $7935$  & $7.8$ \\
\ref{exp:safe5} & 6 & $-$    & $763$   & $1.9$ \\
\ref{exp:safe6} & 8 & $-$    & $332$ & $2.2$ \\
\ref{exp:safe_lorenz12} & 12 & $-$   & $-$ & $4.9$ \\
\bottomrule
\end{tabular}}
\label{tab:safe}
\end{table}


As evidenced by Table \ref{tab:safe}, the SOS-based method (Eq. \eqref{eq:sos_ra_exp_s}) verifies all benchmark instances successfully, while Fossil fails to complete verification within the time limit for Examples \ref{exp:safe2}–\ref{exp:safe_lorenz12}. The SOS-based method  also outperforms DeepICBC in computational efficiency, requiring substantially less time. Notably, SOS (Eq. \eqref{eq:sos_ra_exp_s}) is the sole method to verify Example \ref{exp:safe_lorenz12}—a 12-dimensional system—in just 4.9 seconds, underscoring its superior scalability for high-dimensional systems.

\section{Conclusion}
\label{sec:conclusion}

This paper addressed  the computation of CRASs for discrete-time polynomial systems. We first establish an equivalence between CRASs and 0-reach-avoid sets in a probabilistic context by incorporating probabilistic considerations into control inputs. Building upon this formulation, which helps eliminate nonlinear terms arising from the coupling of Lyapunov-like functions and control inputs in existing conditions for computing CRASs, we leverage convex optimization techniques to compute 0-reach-avoid sets to estimate CRASs. Furthermore, inspired by the $\epsilon$-greedy strategy in reinforcement learning, we proposed an iterative approach that updates probability distributions imposed on control inputs and Lyapunov-like functions, to minimize conservatism in the estimating CRASs. Finally, we tested our approach on extensive numerical examples and provided a comprehensive comparison with existing methods. 


\bibliographystyle{splncs04}
\bibliography{reference}

\renewcommand\thesubsection{\Alph{subsection}}

\section*{Appendix}

\subsection{SOS Programming Implementation}
\label{app:sos}

The optimization problems \ref{eq:updata_u_2} and \ref{eq:v_2} can be encoded as SOS programs, as presented in \eqref{sos:u} and \eqref{sos:v_2}, respectively.

\begin{algorithm}
\begin{equation}
\label{sos:u}
\begin{split}
&\min \sum_{i=0,1,\ldots,D} \|\tilde{\bm{u}}_0(\bm{a},\bm{x}_i) - \bm{u}_i \|_2^2 \\
&{\rm s.t. ~}
\begin{cases}
\tilde{\bm{u}}_0(\bm{a},\bm{x})(1) - (\underline{\bm{u}}(1) + \delta) + s_{1,1}(\bm{x})h(\bm{x})-s_{1,2}(\bm{x})g(\bm{x})\in \sum[\bm{x}],\\
(\overline{\bm{u}}(1)-\delta) - \tilde{\bm{u}}_0(\bm{a},\bm{x})(1) +  s_{2,1}(\bm{x})h(\bm{x})-s_{2,2}(\bm{x})g(\bm{x})\in \sum[\bm{x}],\\
\vdots\\
\tilde{\bm{u}}_0(\bm{a},\bm{x})(m) - (\underline{\bm{u}}(m) + \delta) + s_{2m-1,1}(\bm{x})h(\bm{x})-s_{2m-1,2}(\bm{x})g(\bm{x})\in \sum[\bm{x}],\\
(\overline{\bm{u}}(m) -\delta) - \tilde{\bm{u}}_0(\bm{a},\bm{x})(m) +  s_{2m,1}(\bm{x})h(\bm{x})-s_{2m,2}(\bm{x})g(\bm{x})\in \sum[\bm{x}],\\
s_{i,1}, s_{i,2}\in \sum[\bm{x}], \forall i \in 1,\ldots,2m,
\end{cases}
\end{split}
\end{equation}
where $\bm{u}_i$ is calculated by optimization \eqref{eq:max_u_1}.
\end{algorithm}

\begin{algorithm}
\begin{equation}
\label{sos:v_2}
\begin{split}
&\max \bm{c}_v \cdot \hat{\bm{w}}_v\\
&{\rm s.t. }
\begin{cases}
\frac{1-\epsilon}{Z}\int_{\tilde{\bm{u}}_0(\bm{x})-\epsilon}^{\tilde{\bm{u}}_0(\bm{x})+\epsilon} v(\bm{f}(\bm{x},\bm{u})) d\bm{u} - \frac{\epsilon}{Z}\int_{\underline{\bm{u}}}^{\overline{\bm{u}}} v(\bm{f}(\bm{x},\bm{u})) d\bm{u}-\lambda v(\bm{x})\\
\qquad\qquad\qquad\qquad\qquad\quad+s_1(\bm{x})h(\bm{x})-s_2(\bm{x})g(\bm{x})\in \sum[\bm{x}],\\
-v(\bm{x})+s_3(\bm{x})\hat{h}(\bm{x})-s_4(\bm{x})h(\bm{x})\in \sum[\bm{x}],\\
s_1(\bm{x}), s_2(\bm{x}), s_3(\bm{x}),  s_4(\bm{x}) \in \sum[\bm{x}],
\end{cases}
\end{split}
\end{equation}
where $\bm{c}_v \cdot \hat{\bm{w}}_v = \int_\mathcal{X} v(\bm{x}) d\bm{x}$ and $Z = (1-\epsilon)(2\delta)^m + \epsilon \prod_{j=1}^m (\overline{\bm{u}}(j) - \underline{\bm{u}}(j))$.

\end{algorithm}

\subsection{Applications to Safety of Discrete-time Systems}
\label{app:cbf}

For the discrete-time system \eqref{eq:systems}, the safety problem is defined as follows.

\begin{problem}
    Given a set of interest $\mathcal{X} \subseteq \mathbb{R}^n$, an initial set $\mathcal{X}_I \subseteq \mathcal{D}$, and an unsafe set $\mathcal{X}_U \subseteq \mathcal{D}$ with $\mathcal{X}_U \cap \mathcal{X}_I = \emptyset$, the safety verification problem is to verify that the system \eqref{eq:systems} is safe, i.e., starting from any state in $\mathcal{X}_I$, there exists a controller that prevents the following event from occurring:
    
     there exists $t\in \mathbb{N}$ such that the system $\eqref{eq:systems}$ enters the unsafe set $\mathcal{X}_U$ at tine step $N$ while staying within $\mathcal{D}$ before $t$.
\end{problem}

CBCs offer a mathematical framework to solve safety verification problem, as shown in Proposition \ref{prop:cbf}.

\begin{proposition}
\label{prop:cbf}
Given $\lambda \in (0, 1)$, if there exists a function $B(\bm{x}): \mathcal{D}\rightarrow \mathbb{R}$ satisfying 
\begin{equation}
    \label{eq:safe_sup}
    \begin{cases}
       \max_{\bm{u}\in\mathcal{U}}B(\bm{f}(\bm{x},\bm{u}))- \lambda B(\bm{x})\geq 0 & \forall \bm{x}\in \mathcal{D},\\ 
       B(\bm{x})\leq 0 & \forall \bm{x}\in \mathcal{X}_U,\\
       B(\bm{x}) > 0 & \forall \bm{x}\in \mathcal{X}_I,\\
    \end{cases}
\end{equation}
then the system \eqref{eq:systems} is safe.
\end{proposition}


Akin to \eqref{eq:ra_exp}, by replacing the maximum operator in \eqref{eq:safe_sup} with the expectation operator, a sufficient condition can also be obtained for solving the safety verification problem.

\begin{theorem}
\label{theo:cbf_prob}
Given $\lambda \in (0, 1)$, if there exists a function $B(\bm{x}): \mathcal{D}\rightarrow \mathbb{R}$ satisfying 
\begin{equation}
    \label{eq:safe_exp}
    \begin{cases}
       E[B(\bm{f}(\bm{x},\bm{u}))]- \lambda B(\bm{x})\geq 0 & \forall \bm{x}\in \mathcal{D},\\ 
       B(\bm{x})\leq 0 & \forall \bm{x}\in \mathcal{X}_U,\\
       B(\bm{x}) > 0 & \forall \bm{x}\in \mathcal{X}_I,\\
    \end{cases}
\end{equation}
then the system \eqref{eq:systems} is safe.
\end{theorem}
\begin{proof}
Since $E[B(\bm{f}(\bm{x},\bm{u}))]\leq \max_{\bm{u}\in \mathcal{U}}B(\bm{f}(\bm{x},\bm{u}))$ for $\bm{x}\in \mathcal{D}$, we conclude that if $B(\bm{x}): \mathcal{D}\rightarrow \mathbb{R}$ satisfies \eqref{eq:safe_exp}, it will satisfy \eqref{eq:safe_sup}. According to Proposition \ref{prop:cbf}, we have the conclusion.

\end{proof}

Let $X_I=\{\bm{x}\in \mathbb{R}^n\mid h_I(\bm{x})\leq 0\}$, $\mathcal{D}=\{\bm{x}\in \mathbb{R}^n \mid h(\bm{x})\leq 0\}$, and $\mathcal{X}_U=\{\bm{x}\in \mathbb{R}^n \mid h_U(\bm{x})\leq 0\}$, where $h_I(\bm{x}), h(\bm{x}), h_U(\bm{x})\in \mathbb{R}[\bm{x}]$. The safety verification problem could be addressed via searching for a polynomial function $v(\bm{x})\in \mathbb{R}[\bm{x}]$ satisfying \eqref{eq:safe_exp}. The  search for a polynomial function $v(\bm{x})\in \mathbb{R}[\bm{x}]$ satisfying \eqref{eq:safe_exp} can be encoded as an SOS program \eqref{eq:sos_ra_exp_s}. 
\begin{algorithm}
\begin{equation}
\label{eq:sos_ra_exp_s}
\begin{split}
\begin{cases}
E[v(\bm{f}(\bm{x},\bm{u}))] - \lambda v (\bm{x}) + s_1(\bm{x})h(\bm{x})\in \sum[\bm{x}],\\
-v(\bm{x})+s_2(\bm{x})h_U(\bm{x})\in \sum[\bm{x}],\\
v(\bm{x})+s_3(\bm{x})h_I(\bm{x})\in \sum[\bm{x}],\\
s_1(\bm{x}), s_2(\bm{x}), s_3(\bm{x}) \in \sum[\bm{x}],
\end{cases}
\end{split}
\end{equation}
where $\lambda\in (0,1)$ is a user-specified value. 
\end{algorithm}



\subsection{Experimental Details}
\label{app:exp}
This section shows more experimental details. All computations are run on a machine equipped with an i9-12900H 2.5GHz CPU with 16GB RAM, where the Matlab package YALMIP \cite{lofberg2004yalmip} is employed for SOS decomposition of multivariate polynomials and Mosek 10.1.21 \cite{aps2019mosek} is used to solve the semi-definite programming problem. To ensure numerical stability when solving the semi-definite programming problem, we restrict the coefficients of the unknown polynomials to the interval $[-1000, 1000]$. To calculate the set $\widehat{\mathcal{X}}$ defined in \eqref{eq:x_hat}, we use the semi-definite programming method proposed in \cite{xue2020inner}. In Examples \ref{exp:ra_2d_vanderpol}--\ref{exp:ra_2d_2008}, Examples \ref{exp:ra_2d_2000}--\ref{exp:ra_4d}, Examples \ref{exp:ra_6d}, the number of iterations is set to $K=10$, $K=3$, and $K=2$, respectively. 

In Examples \ref{exp:safe1}--\ref{exp:safe4}, since neither Fossil nor DeepICBC supports control input constraints, all benchmarks in this section exclude such constraints. However, since Alg. \ref{alg} requires a control input set for its computation, we specify a control input set for each example.

Examples \ref{exp:ra_2d_2008}--\ref{exp:safe4} are presented as follows.

\begin{example}
    \label{exp:ra_2d_2008}
    Consider the discretization of the controlled Moore-Greitzer model of a jet engine described in \cite{aylward2008stability},
    \begin{equation*}
        \label{eq:exp2}
        \begin{cases}
        x(t+1)=x(t) + 0.001(-y(t)-1.5x^2(t)-0.5x^3(t)-0.5+u_1(t)),\\
        y(t+1)=y(t) + 0.001(3x(t)-y(t)+u_2(t)),
        \end{cases}
    \end{equation*}
    with the safe set $\mathcal{X}=\{\,(x,y)^{\top}\mid x^2+y^2-1 < 0\,\}$, the target set $\mathcal{T}=\{\,(x,y)^{\top}\mid (x+0.2)^2+(y+0.5)^2-0.05< 0\,\}$, and the control input set $\mathcal{U}=\{\,(u_1,u_2)^\top\mid -1 \leq u_1,u_2 \leq 1\,\}$.
\end{example}

\begin{example}
    \label{exp:ra_2d_2000}
    Consider the discrete-generation predator-prey model \cite{halanay2000stability},
    \begin{equation*}
        \begin{cases}
        x(t+1)= 0.5x(t) - x(t)y(t),\\
        y(t+1)= -0.5y(t) + (u(t) + 1)x(t)y(t),
        \end{cases}
    \end{equation*}
    with the safe set $\mathcal{X}=\{\,(x,y)^{\top}\mid x^2+y^2-1.8^2 < 0\,\}$, the target set $\mathcal{T}=\{\,(x,y)^{\top}\mid x^2+y^2-0.01< 0\,\}$, and the control input set $\mathcal{U}=\{\,u\mid -2 \leq u \leq 2\,\}$.
\end{example}

\begin{example}
    \label{exp:ra_3d_vanderpol} Consider a controlled 3D VanderPol oscillator from \cite{korda2014controller} with a discrete time $0.01$,
    \begin{equation*}
        \begin{cases}
        x(t+1)=x(t) + 0.01(-2y(t)),\\
        y(t+1)=y(t) + 0.01(0.8x(t)-2.1y(t)+z(t)+10x^2(t)y(t)),\\
        z(t+1)=z(t) + 0.01(-z(t)+z^3(t)+0.5u(t)),
        \end{cases}
    \end{equation*}
    with the safe set $\mathcal{X}=\{\,(x,y,z)^{\top}\mid x^2+y^2+z^2-1 < 0\,\}$, the target set $\mathcal{T}=\{\,(x,y,z)^{\top}\mid x^2+y^2+z^2-0.1< 0\,\}$, and the control input set $\mathcal{U} = \{u\mid -5 \leq u \leq 5\}$.
\end{example}

\begin{example}
    \label{exp:ra_4d} Consider a model adapted from \cite{sankaranarayanan2013lyapunov} with a discrete time $0.01$,
    \begin{equation*}
        \begin{cases}
        x_1(t+1)=x_1(t) + 0.01(-x_1(t) + x_2^3(t) - 3x_3(t)x_4(t) + u(t)),\\
        x_2(t+1)=x_2(t) + 0.01(-x_1(t) - x_2^3(t)),\\
        x_3(t+1)=x_3(t) + 0.01(x_1(t)x_4(t)-x_3(t)),\\
        x_4(t+1)=x_4(t) + 0.01(x_1(t)x_3(t)-x_4^3(t)),\\
        \end{cases}
    \end{equation*}
    with the safe set $\mathcal{X}=\{\,(x_1,x_2,x_3,x_4)^{\top}\mid \sum_{i=1}^4 x_i^2-1 < 0\,\}$, the target set $\mathcal{T}=\{\,(x_1,x_2,x_3,x_4)^{\top}\mid \sum_{i=1}^4 x_i^2-0.01 < 0\,\}$, and the control input set $\mathcal{U} = \{u\mid -1 \leq u \leq 1\}$. Since $\int_\mathcal{X} v(\bm{x}) d\bm{x}$ in SOS \eqref{eq:sos_ra_exp} and SOS \eqref{sos:v_2} is computationally expensive in high-dimensional systems, in EX\ref{exp:ra_4d} and EX\ref{exp:ra_6d}, we approximate it by sampling $100$ states $\{\bm{x}_1, \bm{x}_2, \ldots, \bm{x}_{100}\}$ and using $\sum_{i=1}^{100} v(\bm{x}_i)$ as a substitute.
\end{example} 

\begin{example}
    \label{exp:ra_6d} Consider a model adapted from \cite{edwards2024fossil} with a discrete time $0.01$,
    \begin{equation*}
        \begin{cases}
        x_1(t+1)=x_1(t) + 0.01(x_2(t)x_4(t)-x_1^3(t)),\\
        x_2(t+1)=x_2(t) + 0.01(-3x_1(t)x_4(t)-x_2^3(t)),\\
        x_3(t+1)=x_3(t) + 0.01(-x_3(t)-3x_1(t)x_4^3(t)),\\
        x_4(t+1)=x_4(t) + 0.01(-x_4(t)+x_1(t)x_3(t)+u(t)),\\
        x_5(t+1)=x_5(t) + 0.01(-x_5(t)+x_6^3(t)),\\
        x_6(t+1)=x_6(t) + 0.01(-x_5(t)-x_6(t) + x_3^4(t) + u(t)),\\
        \end{cases}
    \end{equation*}
    with the safe set $\mathcal{X}=\{\,(x_1,x_2,\ldots,x_6)^{\top}\mid \sum_{i=1}^6 x_i^2-1 < 0\,\}$, the target set $\mathcal{T}=\{\,(x_1,x_2,\ldots,x_6)^{\top}\mid \sum_{i=1}^6 x_i^2-0.01 < 0\,\}$, and the control input set $\mathcal{U} = \{u\mid -1 \leq u \leq 1\}$.
\end{example} 


\begin{example}
    \label{exp:safe1} This example is taken directly from \cite{edwards2024fossil}, which models a discrete-time system for temperature regulation in a circular building with two rooms, as defined below.
    \begin{equation*}
        \begin{cases}
            &x(t+1) = 1 - \tau (\alpha + \alpha_{e1}) x(t) + \tau \alpha y(t) + \tau \alpha_{e_1} T_e + \tau \alpha_h (T_h - x(t)) u_1(t),\\
            &y(y+1) = 1 - \tau (\alpha + \alpha_{e2}) y(t) + \tau \alpha x(t) +  \tau \alpha_{e_2} T_e + \tau \alpha_h (T_h - y(t)) u_2(t),
        \end{cases}
    \end{equation*}
    with the discretization parameter $\tau = 5$, heat exchange parameters $\alpha = 0.05$, $\alpha_{e_1} = 0.005$, and $\alpha_{e_2} = 0.008$, the external temperature $T_e = 15$, heat exchange parameters between room and heater $\alpha_h = 0.0036$, and the boiler temperature $T_h= 55$. The state space is $\mathcal{D}=\{\,(x,y)^{\top}\mid 17 \leq x,y \leq 30\,\}$, the initial set is $\mathcal{X}_I=\{\,(x,y)^{\top}\mid 17 \leq x, y \leq 18\,\}$, the unsafe set is $\mathcal{X}_U=\{\,(x,y)^{\top}\mid 28 \leq x, y \leq 30\,\}$. When solving Alg. \ref{alg}, we use the control input set $\mathcal{U} = \{(u_1, u_2)\mid -100 \leq u_1, u_2 \leq 100\}$.
\end{example} 

\begin{example}
    \label{exp:safe2} This example is identical to Example 1 in \cite{ren2024formal}, featuring the same initial set, unsafe set, and safe set. In the implementation of Algorithm \ref{alg}, we define the control input set as $\mathcal{U} = \{u \mid -1 \leq u \leq 1\}$.
\end{example} 

\begin{example}
    \label{exp:safe3} This example is identical to Example 4 in \cite{ren2024formal}, featuring the same initial set, unsafe set, and safe set. In the implementation of Algorithm \ref{alg}, we use the control input set $\mathcal{U} = \{u\mid -1 \leq u \leq 1\}$.
\end{example} 

\begin{example}
    \label{exp:safe4} This example is identical to Example 5 in \cite{ren2024formal}, featuring the same initial set, unsafe set, and safe set. In the implementation of Algorithm \ref{alg}, we use the control input set $\mathcal{U} = \{u\mid -100 \leq u \leq 100\}$.
\end{example} 

\begin{example}
    \label{exp:safe5} Consider a 6-dimensional system adapted from \cite{peruffo2021automated},
    \begin{equation*}
        \begin{cases}
        x_i(t+1)=x_i(t) + 0.01x_{i+1}(k), i=1,\ldots,5\\
        x_6(t+1)=x_6(t) - 0.01(576x_1(t)+2400x_2(t)+4180x_3(t)+3980x_4(t)\\
        \qquad\qquad\qquad\qquad\qquad+2273x_5(t)+800x_6(t)-u(t)),
        \end{cases}
    \end{equation*}
    with the state space $\mathcal{D}=\{\,(x_1,\ldots,x_6)^{\top}\mid \sum_{i=1}^6 x_i^2-25 < 0\,\}$, the initial set $\mathcal{X}_I=\{\,(x_1,\ldots,x_6)^{\top}\mid \sum_{i=1}^6 (x_i-1)^2-0.1 < 0\,\}$, and the unsafe set $\mathcal{X}_U=\{\,(x_1,\ldots,x_6)^{\top}\mid \sum_{i=1}^6 (x_i+1.8)^2-0.1 < 0\,\}$. When solving Alg. \ref{alg}, we use the control input set $\mathcal{U} = \{u\mid -1000 \leq u \leq 1000\}$.
\end{example} 

\begin{example}
    \label{exp:safe6} Consider an 8-dimensional system adapted from \cite{peruffo2021automated},
    \begin{equation*}
        \begin{cases}
        x_i(t+1)=x_i(t) + 0.01x_{i+1}(t), i=1,\ldots,7\\
        x_8(t+1)=x_8(t) - 0.01(576x_1(t)+2400x_2(t)+4180x_3(t)+3980x_4(t)\\
        \qquad\qquad\qquad\qquad+2273x_5(t)+800x_6(t)+170x_7(t)+20x_8(t)-u(t)),
        \end{cases}
    \end{equation*}
    with the state space $\mathcal{D}=\{\,(x_1,\ldots,x_8)^{\top}\mid \sum_{i=1}^8 x_i^2-32 < 0\,\}$, the initial set $\mathcal{X}_I=\{\,(x_1,\ldots,x_8)^{\top}\mid \sum_{i=1}^8 (x_i-1)^2-0.1 < 0\,\}$, and the unsafe set $\mathcal{X}_U=\{\,(x_1,\ldots,x_8)^{\top}\mid \sum_{i=1}^8 (x_i+1.8)^2-0.1 < 0\,\}$. When solving Alg. \ref{alg}, we use the control input set $\mathcal{U} = \{u\mid -1000 \leq u \leq 1000\}$.
\end{example} 

\begin{example}
    \label{exp:safe_lorenz12} Consider a 12-dimensional Lorenz model adapted from \cite{lorenz1996predictability},
    \begin{equation*}
        \begin{cases}
        x_i(t+1)=x_i(t) + 0.01((x_{i+1}(t)-x_{i-2}(t))x_{i-1}(t)-x_i(t)), i=1,\ldots,11\\
        x_{12}(t+1)=x_{12}(t) + 0.01((x_{1}(t)-x_{10}(t))x_{11}(t)-x_{12}(t) + u(t)),
        \end{cases}
    \end{equation*}
    where $x_{-1}(t) = x_{11}(t)$ and $x_{0}(t) = x_{12}(t)$. The state space is  $\mathcal{D}=\{\,(x_1,\ldots,x_{12})^{\top}\mid \sum_{i=1}^{12} x_i^2-16 < 0\,\}$, and the initial set is $\mathcal{X}_I=\{\,(x_1,\ldots,x_{12})^{\top}\mid \sum_{i=1}^{12} (x_i-1)^2-0.1 < 0\,\}$. The unsafe set $\mathcal{X}_U=\{\,(x_1,\ldots,x_8)^{\top}\mid 12 < \sum_{i=1}^8 x_i^2 < 16\,\}$. When solving Alg. \ref{alg}, we use the control input set $\mathcal{U} = \{u\mid -1 \leq u \leq 1\}$.
\end{example} 

\end{document}